\documentclass[english,acmsmall, screen, nonacm]{acmart}
\usepackage{xcolor}
\usepackage{array}
\usepackage{fancybox}
\usepackage{calc}
\usepackage{mathtools}
\usepackage{multirow}
\usepackage{amsmath}
\usepackage{amsthm}
\usepackage{graphicx}
\usepackage{subfig}
\usepackage{resizegather}
\usepackage{wrapfig}
\usepackage{caption}
\usepackage{tikz}
\usetikzlibrary{positioning,fit,calc,arrows.meta, automata}
\usepackage{tcolorbox}
\usepackage{enumerate}
\usepackage{algorithm2e}
\usepackage{listings}
\usepackage{graphicx}
\usepackage{etoolbox}

\makeatletter

\theoremstyle{plain}
\newtheorem{thm}{\protect\theoremname}
\theoremstyle{definition}
\theoremstyle{remark}
\newtheorem{rem}[thm]{\protect\remarkname}

\usepackage{etoolbox}
\usepackage{url}
\usepackage[obeyFinal,textwidth=40px,textsize=footnotesize]{todonotes}
\usepackage{multirow}
\newrobustcmd{\tup}[1]{\langle #1 \rangle}
\newrobustcmd{\scrsf}[1]{\text{\scriptsize{\ensuremath{\mathsf{\uppercase{#1}}}}}}
\newrobustcmd{\ns}[1]{\mathsf{NS}_{#1}}
\newrobustcmd{\rec}[1]{\mathsf{Res}_{#1}}
\newrobustcmd{\dn}[1]{\todo[color=yellow]{#1}}
\newrobustcmd{\dni}[1]{\textcolor{magenta}{[DN: {#1}]}}
\newrobustcmd{\swap}{\mathsf{swap}}
\newrobustcmd{\drop}{\mathsf{drop}}
\newrobustcmd{\save}{\mathsf{save}}
\newrobustcmd{\push}{\mathsf{push}}
\newrobustcmd{\pop}{\mathsf{pop}}
\newrobustcmd{\nop}{\mathsf{nop}}
\newrobustcmd{\ec}[1]{{#1}_{\sim_{\cal{Z}}}}
\newrobustcmd{\cal}{\mathcal}
\newrobustcmd{\merge}{\!\mid\!}
\newrobustcmd{\ecmerge}{\!\mid\!\mid\!}
\newrobustcmd{\zoneAlph}{\Sigma_{_{\cal Z}}}
\newrobustcmd{\zoneAlphCompl}{other}
\newrobustcmd{\dnsAlph}{\Sigma_{_{\scrsf{dns}}}}
\newrobustcmd{\simZ}{\sim_{_{\cal Z}}}
\newrobustcmd{\synhom}[1]{\eta_{#1}}
\newrobustcmd{\synmon}[1]{\scrsf{syn}(#1)}
\newrobustcmd{\synhomZ}{\synhom{_{_{\cal Z}}}}
\newrobustcmd{\synmonZ}{\scrsf{syn}_{_{\cal Z}}}
\newrobustcmd{\dnsZ}{\scrsf{dns}_{_{\cal Z}}}
\newrobustcmd{\ch}{\mathsf{ch}}
\newrobustcmd{\trans}{\scrsf{trans}}
\newrobustcmd{\subsub}[1]{_{_{#1}}}
\newrobustcmd{\reach}{\scrsf{reach}}
\newrobustcmd{\rat}{\mathbb{Q}}
\newrobustcmd{\nat}{\mathbb{N}}
\newrobustcmd{\config}[1]{\scrsf{config}(#1)}
\newrobustcmd{\monoid}{\cal M}
\newrobustcmd{\abstraction}{\scrsf{dns}\subsub{Z(\cal Z)}}
\newcommand{\ostar}{\mathbin{\mathpalette\make@circled\star}}
\newcommand{\make@circled}[2]{%
	\ooalign{$\m@th#1\smallbigcirc{#1}$\cr\hidewidth$\m@th#1#2$\hidewidth\cr}%
}
\newcommand{\smallbigcirc}[1]{%
	\vcenter{\hbox{\scalebox{0.77778}{$\m@th#1\bigcirc$}}}%
}
\renewcommand{\ker}[1]{\text{ker}(#1)}
\lstdefinestyle{dnszonestyle}{ backgroundcolor={},
	basicstyle=\ttfamily\scriptsize, 
	basicstyle=\sffamily\scriptsize, breakatwhitespace=false,
	breaklines=true,
	captionpos=b,
	keepspaces=true,
	numbers=none, %
	showspaces=false,
	showstringspaces=false, showtabs=false,
	tabsize=2, moredelim=**[is][\color{orange}]{@}{@}, }
\lstdefinelanguage{DNSZone}{
	morecomment=[l]{---} %
}
\newrobustcmd{\bottom}{\!\!\!\dashv}
\newrobustcmd{\verif}{\scrsf{dnsverif}}

\usepackage{booktabs}
\usepackage{xspace}
\usepackage{cleveref}
\usepackage{pifont}

\newcommand{\inlsec}[1]{\smallskip\noindent\textbf{#1.}}

\newcommand{\circled}[1]{{{\textcircled{\small #1}}}}

\newcommand{\new}[1]{{{#1}}}

\newcommand{\name}[1]{\texttt{#1}\xspace}

\newcommand{\rtype}[1]{\MakeLowercase{\textsc{#1}}}

\makeatother

\usepackage{babel}

\providecommand{\remarkname}{Remark}
\providecommand{\theoremname}{Theorem}

\begin{document}
	\title{Reachability Analysis of the Domain Name System}

	\author{Dhruv Nevatia}
	\email{dhruv.nevatia@inf.ethz.ch}
	\orcid{0009-0008-0845-6754}
	
	\author{Si Liu}
	\email{si.liu@inf.ethz.ch}
	\orcid{0003-3578-7432}
	
	\author{David Basin}
	\email{basin@inf.ethz.ch}
	\orcid{0003-2952-939X}
	
	\affiliation{
		\institution{ETH Zurich}
		\country{Switzerland}
	}
	
	\begin{abstract}
The high complexity of DNS poses unique challenges for ensuring  its security and reliability. Despite continuous advances in DNS testing, monitoring, and verification, protocol-level defects still give rise to numerous bugs and attacks. In this paper, we provide the first decision procedure for the DNS verification problem, establishing its complexity as $\mathsf{2ExpTime}$, which was previously unknown. 

We begin by formalizing the semantics of DNS as a system of recursive communicating processes extended with timers and an infinite message alphabet. We provide an algebraic abstraction of the alphabet with finitely many equivalence classes, using the subclass of semigroups that recognize positive prefix-testable languages. We then introduce a novel generalization of bisimulation for labelled transition systems, weaker than strong bisimulation, to show that our abstraction is sound and complete. Finally, using this abstraction, we reduce the DNS verification problem to the verification problem for pushdown systems. To show the expressiveness of our framework, we model two of the most prominent attack vectors on DNS, namely amplification attacks and rewrite blackholing.
\end{abstract}


	
	\maketitle
	
	\section{Introduction} \label{sec:intro}

The Domain Name System (DNS) is a central component of the Internet's infrastructure. 
It translates human-readable domain names, such as \name{www.sigplan.org}, into machine-recognizable IP addresses via \emph{name resolution}, thereby simplifying Internet navigation and resource access for users. 
This seemingly simple translation service,  handled by \emph{recursive resolvers}, is underpinned by an intricate, hierarchical, globally distributed database. 
Each organization, like Google or Cloudflare, 
provides name resolution for its portion or \emph{zone} of the entire DNS namespace.
Operators within each organization manage that zone through,  often manually configured, \emph{zone files}. 
These zone files, 
stored on \emph{authoritative nameservers},
map domain names to IP addresses as well as other types of DNS records, 
specifying further actions to be taken such as query rewriting or delegation. 

\new{The IETF (Internet Engineering Task Force)~\cite{ietf} publishes \emph{requests for comments} (RFCs), which are technical documents describing the Internet's technical foundations.}
To date, hundreds of RFCs (e.g., \cite{rfc1034,rfc1035,rfc6672,rfc2181}) have been published that define and guide the design, implementation, and operation of DNS. 
The sophisticated nature and large scale of DNS pose unique challenges for developers and operators, who have invested considerable efforts to ensure its functional correctness,   security, and availability.
These challenges  often stem from name resolution failures, which result mainly from misconfigurations or  attacks, and have historically led to large-scale outages~\cite{DNS-outage-1,facebook-outage,azure-outage}. Of particular concern are \emph{denial of service} (DoS)  vectors, which have been frequently identified over the past decade~\cite{iDNS,NRDelegationAttack,TUDOOR}. In particular,  \emph{amplification attacks}, which exploit the DNS protocol to significantly 
increase the query load on the victim (resolvers or authoritative nameservers), have seen a surge in recent years~\cite{NXNSAttack,dnsmaude,TsuKing,TsuNAME,camp}.

Despite continuous advances in DNS testing and monitoring~\cite{ResolFuzz,ResolverFuzz,thousandeyes,check-host,Eywa,scale}, 
numerous issues, like those just mentioned, still arise from RFCs or protocol-level defects in extensively tested production resolvers and nameservers. This highlights the necessity for 
a principled, proactive way to mitigate such problems, 
preferably at an early design stage. 
Two recent formal efforts~\cite{dnsmaude,groot} have attempted to address this issue. 
GRoot~\cite{groot} is the first static verifier for DNS configuration errors, building on the protocol-level DNS resolution semantics.
It has been successfully applied to identify 
misconfigurations in DNS zone files, prior to their deployment. 
DNSMaude~\cite{dnsmaude} is a formal framework, implemented in the Maude~\cite{DBLP:conf/maude/2007} language, 
for both the qualitative (e.g., functional correctness) and quantitative (e.g., amplification) analysis of  DNS protocols. 
It is based on a 
semantics for the entire end-to-end name resolution, 
 covering essential features such as resolver caching and recursive subqueries, which GRoot abstracts away.
\emph{Henceforth, when we refer to DNS, we consider this more comprehensive semantics.}
DNSMaude's accompanying analyzer has also discovered multiple attacks on DNS with large amplification effects. %

\inlsec{Research Gaps} %
Both works have provided partial solutions to the \emph{DNS verification}   
(\scrsf{dnsverif}) problem
where, 
given the DNS zone files $\cal Z$,
one asks whether a property $\phi$ of interest holds for some behavior of DNS (see \Cref{sec:property-validation-to-reachability}  for its formal definition). 
Intuitively, the property $\phi$ captures the \emph{bad} behaviors exhibiting attacks. As every behavior of DNS starts with an initial query from the resolver, on behalf of a client, an immediate challenge in this problem is to explore the infinite space of queries exhaustively. A priori, this was infeasible.

GRoot tackles this problem by introducing an equivalence relation on the query space, with finitely many equivalence classes (ECs), such that queries within an EC are resolved in the same manner and yield the same result. While this equivalence is proven to be sound under a simplified DNS semantics, it is unsound with respect to the more realistic semantics involving resolver caching (formalized in DNSMaude). Moreover, GRoot must be supplied with a pre-determined bound on the length of runs in DNS, making it incomplete by design. Additionally, the number of ECs increases explosively with query rewrites, to the order of $n^{255}$ in the size of  zone files, due to advanced DNS features like \rtype{DNAME} rewriting \cite{rfc6672} (see also \Cref{subsec:dns}), which can quickly inflate GRoot's verification time. Finally, GRoot is inherently incapable of specifying behavioural vulnerabilities like amplification DoS attacks within its semantic framework.

DNSMaude also works with GRoot's ECs but supports specifying common behavioural vulnerabilities. However, it still struggles with 
attack discovery. Since it uses GRoot's ECs, it may fail to explore certain bad behaviours that it believes are equivalent to those that have already been explored. It may then incorrectly conclude that there is an absence of bad behaviours, rendering its search procedure incomplete.
\new{
Moreover, DNSMaude's search procedure is also incomplete by design; it uses bounded, explicit-state, linear temporal logic (LTL) model checking~\cite{DBLP:conf/maude/2007}
for functional correctness properties, searching for bugs or potential attacks. 
DNSMaude also employs
 statistical model checking (SMC)~\cite{Sen05} for quantitative properties like amplification. 
 SMC verifies a property probabilistically, up to a given confidence level, by running Monte Carlo simulations. As a result, the search for attacks is again bounded, for example, in the number of simulations, and attacks may be missed.\footnote{\new{We have experimented with DNSMaude's simulation mode, which is used to discover attacks as presented in~\cite{dnsmaude}. We observe that around 30\% of the runs do not achieve the maximum amplification factor, thus potentially overlooking attacks.}} 
 }

Overall, neither of these two efforts provides a \emph{decision procedure} for the $\scrsf{dnsverif}$ problem. Moreover, an upper bound on the complexity of this problem still remains unknown~\cite{dns-complex}.

\inlsec{Our Approach} In this paper, we aim to address the above challenges. 
As the first step, we must faithfully model the semantics of DNS. We consider the well-studied model of a system of recursive communicating processes \cite{heussner2010reachability,heussner2012reachability}, and extend it with timers. We call this extension, a \emph{system of recursive communicating processes with timers} (trCPS). We  model the DNS semantics as an instance of a trCPS. 
Intuitively, the timer keeps track of the age of a resource record in a resolver’s cache. 
As resolvers are \emph{stateful} (due to referrals or query rewriting),
we utilize stacks to model how they  track the resolution process, where  each stack entry corresponds to a subquery being resolved. 
Nameservers just behave as \emph{labelled transition systems} (LTSs) as they are \emph{stateless}.

\paragraph{DNS is Eager}
It is well-known that reachability is undecidable for networks of LTSs that contain cycles in the underlying topology, assuming a finite message alphabet and perfect channels \cite{brand1983communicating}.
However, this problem becomes  decidable with \emph{lossy communications}~\cite{finkel2001well, AbJo:lossy:IC}. 
Since DNS typically operates over UDP (User Datagram Protocol)~\cite{rfc1035}, which is inherently lossy, we can use lossy channels.
 \new{Nevertheless, the problem once again becomes undecidable if a cycle contains a process that is a pushdown system (PDS)~\cite[Thm. 1]{aiswarya2020network}.}
 Unfortunately, DNS falls within this class, as a resolver is modeled as a PDS with timers and communicates bidirectionally with nameservers (sending queries and receiving answers). 
 Heussner et al.~\cite{heussner2010reachability, heussner2012reachability} show that the reachability problem for networks of PDSs over \emph{perfect} channels \new{and a \emph{finite} message alphabet}, is decidable for eager runs and pointed network topologies. We show that all runs of DNS are eager by design, and that the underlying topology is pointed. This makes DNS fall within this decidable class of systems.

\paragraph{Congruence}
 Since DNS operates over both (i) an infinite query space and (ii) infinitely many timer valuations,
 the eagerness of DNS by itself is insufficient to obtain decidability. \new{This is because an infinite query space also necessitates an infinite message alphabet.} We propose a novel \emph{congruence} on the domain space with finitely many equivalence classes, derived using syntactic congruences for carefully constructed regular languages \cite{nerode1958synmon,rabinscott1959synmon}, from the given zone file configuration $\cal Z$. We then show that the domains in any two DNS configurations only need to be syntactically congruent to each other with respect to certain \emph{positive prefix-testable languages} (see \Cref{subsec:ppt}). Moreover, these languages are parametric only in the names and values appearing in $\cal Z$'s resource records (see \Cref{subsec:def-not}), in order to exhibit equivalent behaviours with respect to the DNS semantics. Notably,
 owing to the nature of these languages, the equivalence is robust to an unbounded number of \rtype{DNAME} rewrites (unlike GRoot), enabling an unbounded analysis. \new{This effectively yields an abstraction with an \emph{equivalent} finite message alphabet for the system.} Finally, the set of timer valuations are abstracted using a standard region construction.

\paragraph{Kernel Bisimulation}
The domain congruence induces a homomorphism on the semantics
of DNS, allowing us to reduce it to a recursive CPS with  finitely many control states. However, to show that the reduction is both sound and complete, we introduce a novel,  generalized notion of bisimulation, called \emph{generalized kernel bisimulation}. A crucial property of this bisimulation is the pointwise equivalence of runs from equivalent states \emph{up to the equivalence of transitions}. We prove that the homomorphism induces a generalized kernel bisimulation on the DNS semantics such that any run in the abstract semantics corresponds to a class of equivalent runs in the original semantics. 
We use this to establish both soundness and completeness of our abstract model with respect to DNS behavior.

\paragraph{Reduction to Reachability}
As we prove that DNS is an eager trCPS over a pointed network topology,
the reduced abstract DNS is also an eager recursive CPS over the same topology. %
We then reduce the $\scrsf{dnsverif}$ problem to an instance of the reachability problem for PDSs, which is known to be decidable \cite{bouajjani1997reachability}.
Thus, we obtain the decidability of this problem. In addition, we establish that it is solvable in doubly-exponential time in the size of the zone file configuration $\cal Z$.

\inlsec{Contributions} Overall, we make the following  contributions.

\begin{itemize}
    \item We establish the first decision procedure for the $\scrsf{dnsverif}$ problem (Theorem~\ref{thm:sound-complete-abstraction}). Moreover, we show that our algorithm has an upper-bound of $\mathsf{2ExpTime}$ complexity (Lemma~\ref{lem:complexity}), which was previously unknown. %

    \item At the technical level, we initially propose trCPS, a model for networks of recursive communicating processes with timers, as the underlying formal model for DNS (Section~\ref{formal-model}). But it comprises an infinite message alphabet due to the infinite query space. \new{Owing to the monoidal nature of queries, we devise an algebraic abstraction for them (\Cref{subsec:syn-domain-congruence}), with finitely many equivalence classes, using a special class of semigroups (Lemma~\ref{lem:prefix-prefix-singleton-ec}).} We then introduce a novel generalization of bisimulation for LTSs, weaker than strong bisimulation \cite{milner1980calculus} but incomparable to weak bisimulation, \new{to show that our abstraction is sound and complete (\Cref{subsec:gen-ker-bis})}. Both these tools may also be of independent theoretical interest.
    \item We show how our framework can be applied to instantiate two of the most prominent attacks on DNS, namely \emph{amplification attacks} and \emph{rewrite blackholing} (Section~\ref{sec:applications}). This demonstrates the versatility of our approach. %
\end{itemize}

	\section{Preliminaries}

\subsection{DNS in a Nutshell} \label{subsec:dns}
The Domain Name System (DNS) is a decentralized database that translates human-readable domain names into  IP addresses. For example, when a user visits the website \name{www.sigplan.org}, the web browser  sends a DNS request to retrieve the IP address of the server hosting the site.

\paragraph{Namespace} 
The DNS namespace is organized hierarchically in a tree structure, 
with a root domain at the top, followed by top-level domains (TLDs), second-level domains (SLDs), and further subdomains. 
\Cref{fig:namespace} shows a small portion of this namespace, including the 
TLDs \name{org} and \name{com} as well as  SLDs such as 
\name{sigplan}. 
A domain name such as \name{www.sigplan.org} uniquely identifies a path in this tree, with the associated labels concatenated  bottom up, using the `.' character; the root node has an empty label, which is typically omitted. 
The entire namespace is divided into \emph{zones} or subtrees,
each managed by authoritative \emph{nameservers}.

\paragraph{Resource Records}
A DNS \emph{resource record} (RR) stores 
specific information about a domain name,
including the owner name (where this RR resides), record type, time to live (TTL), and 
value. In general, multiple
RRs with the same owner name and type can exist, provided they have
different values. 
Common types of RRs include 
\rtype{A}, \rtype{NS}, \rtype{SOA}, \rtype{TXT},  \rtype{CNAME}, and \rtype{DNAME}, 
each serving a distinct purpose. 
For example, an \rtype{A} record 
maps a domain name to an IPv4 address.
An \rtype{NS} record
specifies the authoritative nameservers for a domain.
A \rtype{CNAME} record
maps an alias domain name, e.g., \name{example.net}, 
to its canonical name, e.g., \name{www.example.com}; any request for \name{example.net} will be redirected to \name{www.example.com}.
A \rtype{DNAME} record
provides redirection for an \emph{entire subtree} of the domain name space to another domain.
For instance, 
a \rtype{DNAME} record at \name{legacy.example.com} with value \name{new.example.com}
would cause a query for \name{any.legacy.example.com} 
to be restarted for \name{any.new.example.com}.

A \emph{wildcard} record is
a special type (\name{$\star$}) of DNS record that
 matches requests for non-existent domain names. 
 This simplifies DNS management by allowing a single record to cover multiple potential queries.
For instance, suppose there is an \rtype{A} record
with \name{1.2.3.4} for \name{$\star$.example.org}.
Queries for  \name{one.example.org} or \name{any.example.org}
will all resolve to \name{1.2.3.4}, 
provided that no other records exist for these specific names.

\paragraph{Name Resolution}
A client  such as a web browser
initiates a DNS request by
directing the query to a \emph{recursive resolver} 
that handles the resolution process on the client's behalf.
If the recursive resolver has a cached response from a prior query, it can promptly reply to the client.
Otherwise,
it queries a hierarchy of authoritative nameservers to retrieve the answer.

\begin{figure}
\begin{minipage}[b]{.33\linewidth}
\centering
   \includegraphics[width=\textwidth]{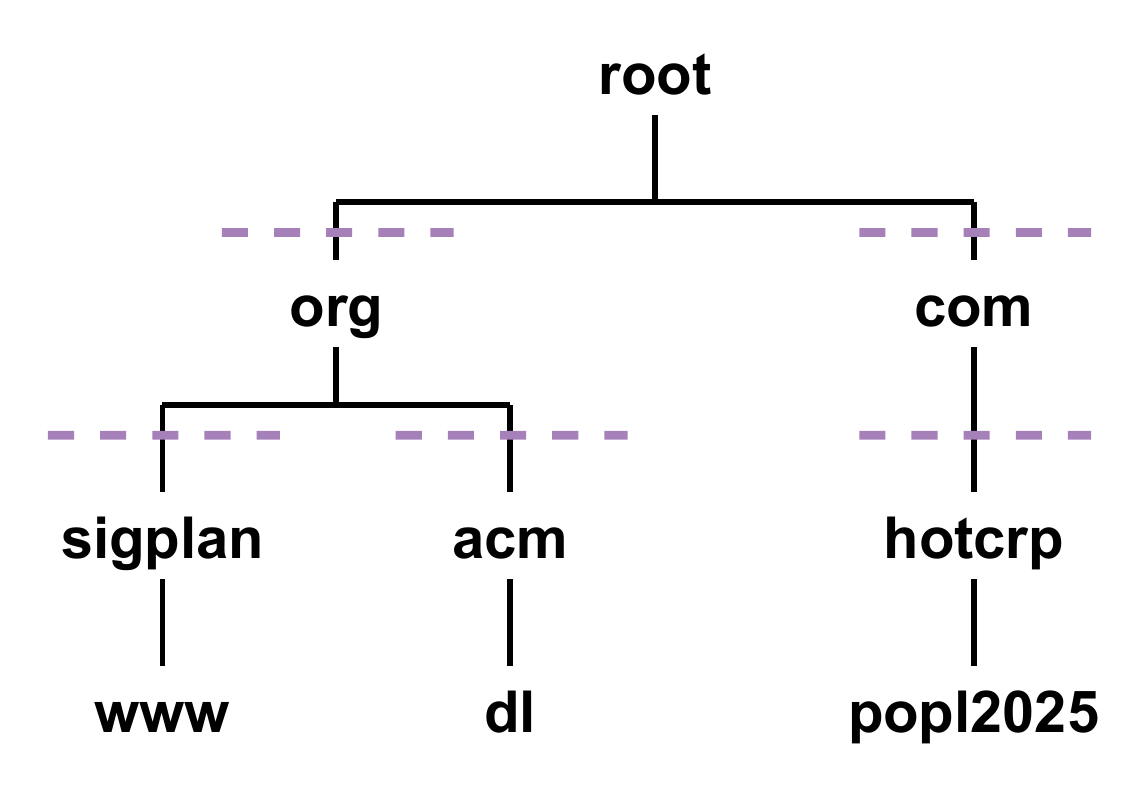}
\captionsetup{skip=10pt}
   \caption{A small portion of the DNS namespace. 
   Dashed lines mark the zone cuts. }
   \label{fig:namespace}
\end{minipage}
\hfill
\begin{minipage}[b]{.65\linewidth}
\centering
   \includegraphics[width=\textwidth]{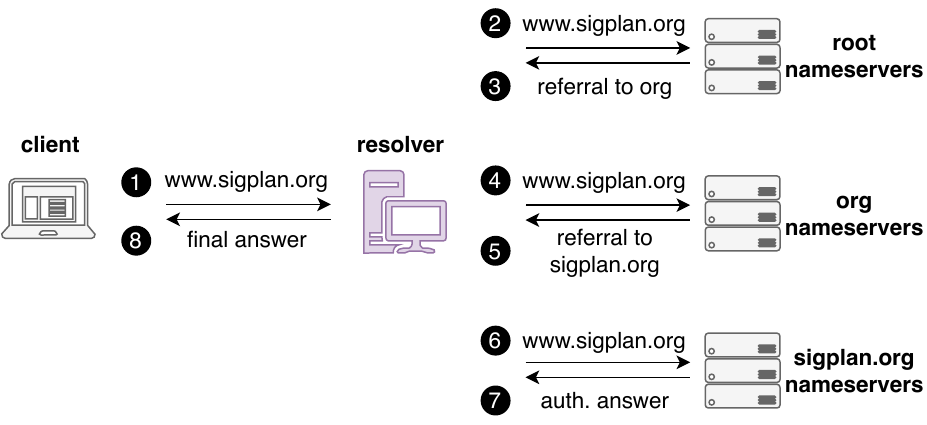}
\captionsetup{skip=0pt}
   \caption{Name resolution for the type \rtype{A} query \name{www.sigplan.org}, with an empty cache at the DNS resolver. }
   \label{fig:dns}
\end{minipage}
\end{figure}

\begin{example}[Name Resolution] \label{example:resolution}
\Cref{fig:dns} depicts a DNS name resolution process. 
A client initiates a type \rtype{A} query for the name \name{www.sigplan.org} (\ding{182}). 
Assuming an empty cache, the recursive resolver forwards this query to one of the root nameservers  (\ding{183}). 
Here, the root nameserver cannot answer the query as it is only authoritative for the \name{root} zone. Hence, it refers the resolver to the authoritative nameservers for the \name{org} zone by  attaching the associated \rtype{NS} records and IP addresses in the referral (\ding{184}). 
The resolver then asks one of the \name{org} nameservers (\ding{185}). As the nameserver is not authoritative for the full domain,  it further refers the resolver to the authoritative nameservers for \name{sigplan.org} (\ding{186}). The resolver  finally receives the result, the IP address for \name{www.sigplan.org} (\ding{187} and \ding{188}), which it then sends back to the client (\ding{189}).
\end{example}

Note that, owing to  \emph{caching},
if the same query needs to be resolved again, the resolver can answer it directly without contacting the nameservers, provided the records' TTL has not expired.
Likewise, for a query of another name within an already cached zone, the resolver only needs to reach out to the relevant nameservers rather than starting from the root.
Note also that when a resolver has multiple nameservers to choose from, the selection is typically based on factors like response time and availability. DNS RFCs do not prescribe a specific nameserver selection algorithm, leaving this to the developers’ discretion.

\new{\paragraph{Subqueries}
During DNS resolution, the recursive resolver may need to perform additional queries to collect all the necessary information to resolve a client query. 
For instance,
suppose that a user wants to visit \name{example.org} and
 the authoritative
nameserver for the \name{example.org} zone is \name{sub.query.com},
the resolver must first resolve \name{sub.query.com}.\footnote{\new{This is a randomly chosen name just for illustrative purposes. As long as the authoritative
nameserver's name  differs from \name{example.org}, the resolver must then make additional subqueries to resolve it.}}
This process begins by querying the \name{com} TLD, followed by \name{query.com} (assuming the authoritative nameserver for \name{query.com} provides the IP address for \name{sub.query.com}). 
 These two additional queries are called \emph{subqueries}, as they are initiated by the resolver to resolve the original query.
}

\subsection{Labelled Transition Systems}\label{subsec:lts}
Given an alphabet $\Sigma$, a \emph{string} $str \in A^*$ is a sequence of labels $l_0 \cdot \ldots \cdot l_k$, where $l_0$ is the first label and $l_k$ is the last label. The length of the sequence, $k+1$, is denoted by $|str|$. For $i < j < |str|$, $str[i\ldots j]$ denotes the sequence $l_i\ldots l_j$ and $\epsilon$ denotes the \emph{empty} sequence. \new{Given two strings $str$ and $str'$, we write $str \preceq str'$ ($str \prec str'$) to denote that $str$ is a (proper) prefix of $str'$.} For notational convenience, we sometimes drop the explicit binary operator for concatenation.

\begin{definition}
	A \emph{labelled transition system} (LTS) \cite{heussner2012reachability}  $\cal A = \tup{S,S_{\mathit{init}},A,\trans}$ is given by a set of \emph{states} $S$, a set of initial states $S_{\mathit{init}}\subseteq S$, an alphabet $A$ of actions, and a set of transition rules $\scrsf{trans}$ of the form $s \xrightarrow{a} s'$, where $s,s'\in S$ and $a \in A$. As usual, we define an \emph{automaton} $\cal F = \tup{S, S_{\mathit{init}}, A, \trans, S_{\mathit{fin}}}$ as the LTS $\cal A$ along with a set of \emph{final} states $S_{\mathit{fin}} \subseteq S$. The automaton $\cal F$ is called a \emph{finite automaton} (FA) if the sets $S$ and $A$ are finite.
\end{definition}

A \emph{run} in the LTS $\cal A$ is a finite sequence $\rho = s_0 \xrightarrow{a_1} s_1 \ldots s_{n-1} \xrightarrow{a_n} s_n$, for $s_i \in S$ such that for every $1 \leq i \leq n$, $s_{i-1} \xrightarrow{a_i} s_i \in \trans$. We say that $\rho$ is a run from $s_0$ to $s_n$. The length of $\rho$ is $n$ and is denoted by $|\rho|$. The \emph{trace} of the run $\rho$ is the sequence of actions $a_1\ldots a_n$. Two runs are \emph{trace equivalent} if they have the same traces. Given the FA $\cal F$, a run in $\cal A$ is \emph{accepting} if it starts in an initial state and ends in a final state.

A state $s \in S$ is \emph{reachable} in $\cal A$ if there exists a run of $\cal A$ from an initial state $s_0$ to $s$.  The \emph{LTS verification} \emph{problem} asks whether, given an LTS $\cal A = \tup{S,S_{\mathit{init}},A,\trans}$, and an automaton $\cal F'$
over strings in $A^*$, there exists a run $\rho$ in $\cal A$ from an initial state, and an accepting, trace equivalent run $\rho'$ in $\cal F'$. We call the corresponding trace a \emph{witness} trace. The problem can be analogously defined for any system with an LTS semantics.

	\section{\label{formal-model} A Formal Model for DNS}

In this section, we define a formal model for DNS and provide the semantics
for all its syntactic components, viz. authoritative nameservers and
recursive resolvers. In \Cref{subsec:trcps}, we first define the meta-model of recursive communicating processes with timers, which is a timed extension of the commonly known model of recursive communicating processes. In  Sections~\ref{subsec:dns-semantics}--\ref{subsec:RR-semantics},
we use our meta-model to comprehensively capture complex features such as recursive subqueries, time-to-live (TTL) records, channel communication, non-determinism, etc. In particular, authoritative nameservers are modelled as LTS components and the recursive resolver is modelled as a tPDS component, allowing it to use timers, for keeping track of the age of cache elements, and a stack, to simulate recursive subquery resolution. As shown in Section \ref{sec:property-agnostic}, we provide a reduction of this model to a recursive CPS.

\subsection{Definitions and Notations}\label{subsec:def-not}
Consider a monoid $\cal M = \tup{M,\cdot, \mathbf{1}}$, where $M = \tup{G}$ is generated by the subset $G \subseteq M$, with the binary operation ($\cdot$), and the identity element $\mathbf{1}$. Given any two subsets $B,C \subseteq M$, we define the sets, $B\cdot C \coloneqq \{x\cdot y\mid x\in B, y\in C\}$, and $m\cdot B \coloneqq \{m\cdot x \mid x\in B\}$, for an arbitrary $m \in M$.

\subsubsection{Domain names and IP addresses}
An IP address is a string $ip\in \scrsf{ip}$, where $\scrsf{ip}$ is the set of all IP addresses, disjoint from $M$. A \emph{domain name} over $\monoid$ is a string $d \in M$ used to point to an IP address. The identity element $\mathbf{1}$ represents the root domain name. Henceforth, as a convention, we write our domain names in the \emph{reverse order}, starting with the labels higher up in the namespace hierarchy. For example, we will write \name{com.example.www} instead of \name{www.example.com}. For better readability, the set of all domain names over $\monoid$ is denoted by $\scrsf{domain}\subsub{\monoid} = M$.

\subsubsection{Resource records}

A \emph{resource record} $rr=\tup{d,t,\sigma,v,b}$ is a tuple consisting
of a \emph{record name} $d\in\scrsf{domain}\subsub{\monoid} \cup \scrsf{domain}\subsub{\monoid_{\star}}$ where
$\scrsf{domain}\subsub{\monoid_{\star}}$ denotes the set of domain names in wildcard records over the isomorphic monoid $\monoid_{\star} = \tup{M_\star = \tup{G_\star},\ostar,\mathbf{1}_{\star}}$ where $G_\star = \{g_\star \mid g\in G\}$, such that $m_{\star}\ostar m_{\star}' = (m \cdot m')_{\star}$; the record type $t\in\scrsf{type}=\{\scrsf{a},\scrsf{ns},\scrsf{mx},\scrsf{dname},\scrsf{cname},\scrsf{soa},\ldots\}$;
a time-to-live value $\sigma\in\mathbb{N}$; the value of the record
$v\in\scrsf{ip}\cup\scrsf{domain}\subsub{\monoid}$;
and a bit value $b\in\{0,1\}$ that denotes whether the record has
been synthesized from a $\scrsf{dname},\scrsf{cname}$, or $\scrsf{wildcard}$
record. We shall write dn$(rr)$, ty$(rr)$, ttl$(rr)$, val$(rr)$, and syn$(rr)$, to denote the fields $d,t,\sigma,v$, and $b$, in $rr$,
respectively. The class of all resource records is denoted by $\scrsf{record}\subsub{\monoid}$.

A \emph{zone} is a set of resource records. The set of all zones is denoted
by $\scrsf{zone}\subsub{\monoid}={\cal P}(\scrsf{record}\subsub{\monoid})$. Additionally, we assume
all zones are \emph{well-formed}, as described in \cite[Appendix A]{groot}. We shall
write $\text{dn}(z)$ to denote the domain name of a well-formed zone.

A \emph{zone file configuration} (ZFC) is a pair, $\tup{\cal Z, \Theta}$, consisting of a function $\cal Z : \cal S \rightarrow \cal P(\scrsf{zone}\subsub{\monoid})$ from a finite set $\cal S \subseteq \scrsf{ip}$ of IPs to the powerset of $\scrsf{zone}\subsub{\monoid}$, along with a subset $\Theta \subseteq \cal S$. Intuitively, $\cal S$ denotes a set of nameservers and $\Theta$ denotes the set of root nameservers in the network; $\cal Z(s)$ denotes the set of zones for which $s$ is authoritative. The set of all ZFCs is denoted by $\scrsf{zfc}\subsub{\monoid}$.

\subsubsection{Queries}

A DNS \emph{query} $q=\tup{d,t}$ is a pair consisting of a domain
name $d\in\scrsf{domain}\subsub{\monoid}$ and the type of query $t\in\scrsf{type}$.
For example, a query expecting a $\rtype{dname}$ rewrite for the domain name
$\mathsf{{com.example}}$ would be $\tup{\mathsf{com.example},\scrsf{dname}}$.
The set of all queries is denoted by $\scrsf{query}\subsub{\monoid}=(\scrsf{domain}\subsub{\monoid}\times\scrsf{type})\uplus\{\bot\}$,
where $\bot$ denotes a non-existent query. We shall write $\text{dn}(q)$
and $\text{ty}(q)$ to denote the fields $d$ and $t$.

\subsubsection{Answers}

\sloppy Every DNS \emph{answer} $a=\tup{l,rec}$ is a pair consisting of
a label $l\in\{\scrsf{ans},\scrsf{ansq},\scrsf{ref},\scrsf{rec},\scrsf{nxdomain},\scrsf{refused}\}$
denoting the type of answer (i.e., an answer $\scrsf{ans}$, a rewrite
$\scrsf{ansq}$, a delegation $\scrsf{ref}$, a recursive subquery
request $\scrsf{rec}$, an indication for a non-existent domain name query $\scrsf{nxdomain}$, etc.) and a set of resource records $\mathit{rec}\in\mathcal{P}(\scrsf{record}\subsub{\monoid})$
that contains information for further evaluation. The set of all answers is denoted by $\scrsf{answer}\subsub{\monoid}$.

\subsubsection{Server List}

A \emph{server list} $sl \in M^*$, is a string of domain names where each domain name is unique. That is, $sl[i] \neq sl[j]$, for $0\leq i < j \leq |sl|$. The set of all server lists is denoted by $\scrsf{slist}\subsub{\monoid}$. We additionally define a left action, $\odot$, of $M$ on $\scrsf{sl}\subsub{\monoid}$ such that for any $m\in M$, 
\[m \odot sl = \begin{cases}
	(m)(sl), & \text{if } m \text{ does not match any position in } sl\\
	(m)(sl'), & \text{otherwise,}
\end{cases}\]
where $sl'$ is the string derived from $sl$ by dropping the (only) position labelled by $m$. Given a string $\mathbf{m} = m_1\ldots m_n \in M^*$, we denote by $\mathbf{m}\odot sl$, the right associative composition of actions $m_1\odot (\ldots \odot(m_n \odot sl))$.

\Cref{fig:DNS-notation} describes some notation we use
when specifying DNS behavior, along with the intended semantics. $\scrsf{rank}$ adopts the product ordering $\vartriangleleft$
over the order $\mathit{false}\leq \mathit{true}$ of truth values. 

\begin{figure}
	\begin{tcolorbox}[width=\linewidth]
        \scalebox{.75}{ 
            \parbox{\linewidth}{ 
			\begin{align*}
				d_1^{-1} \cdot d_2 & \coloneqq \{d \mid d\in M, d_1\cdot d = d_2\} & \text{\emph{domain inverse relation}}\\
				prec(d) & \coloneqq \{d_1 \mid d \in d_1 \cdot M\} & \text{\emph{domain prefix set}}\\
				d_{1}\simeq_{j}d_{2} & \coloneqq |prec(d_1) \cap prec(d_2)| \geq j - 1 & \text{\emph{domain prefix match}}\\
				max_{\simeq}(d_{1},d_{2}) & \coloneqq |prec(d_1) \cap prec(d_2)| - 1 & \text{\emph{maximal prefix match}}\\
				d_{1} < d_{2} & \coloneqq prec(d_1) \subset prec(d_2) & \text{\emph{domain partial order}}\\
				\text{dn}(rr)\ni_{{\star}} d & \coloneqq d_\star \in \text{dn}(rr)\ostar G_\star \wedge \text{dn}(rr)\in M_\star & \text{\emph{domain wildcard match}}\\
				\scrsf{wildcard}(rr) & \coloneqq \text{dn}(rr) \in M_\star & \text{\emph{wildcard record}}\\
				\scrsf{match}(rr,d) & \coloneqq \text{dn}(rr)\ni_{{\star}} d \vee (d \in \text{dn}(rr)\cdot M \wedge \text{dn}(rr)\in M) & \text{\emph{record matches query}}\\
				\scrsf{rank}(rr,\tup{d,t},z) & \coloneqq\tup{\scrsf{match}(rr,d),\text{ty}(rr)=\scrsf{ns}\wedge\text{dn}(rr)\neq\text{dn}(z),max_{\simeq}(\text{dn}(rr),d),\scrsf{Wildcard}(rr)} & \text{\emph{resource record rank}}\\
				rr_1<_{q,z}rr_2 & \coloneqq\scrsf{rank}(rr_1,q,z)\vartriangleleft\scrsf{rank}(rr_2,q,z) & \text{\emph{resource record order}}
			\end{align*}
}}   
	\end{tcolorbox}
	\captionsetup{skip=0pt}
	\caption{\label{fig:DNS-notation}Some DNS notation.} 
\end{figure}

\begin{example}\label{ex:normal-domain}
	Consider the free monoid of strings $\cal L = \tup{\scrsf{label}^*,(.),\epsilon}$, generated by an infinite set of labels $\scrsf{label}$ under concatenation. The strings $\mathsf{com.example.a}$ and $\mathsf{com.example.b}$ are examples of domains in $\scrsf{domain}_{\cal L}$, if the labels $\mathsf{com,example,a,b}\in \scrsf{label}$. The empty string $\epsilon$ represents the root domain. 
	
	The expression $(\mathsf{com}^{-1}.\mathsf{com.example.a})$ denotes the singleton set $\{\mathsf{example.a}\}$; $prec(\mathsf{com.example.a})$ is the set of prefixes $\{\epsilon, \mathsf{com,com.example,com.example.a}\}$; $max_{\simeq}(\mathsf{com.example.a},\mathsf{com.example.b})$ is $2$, the size of the set $\{\mathsf{com,com.example}\}$, of non-trivial common prefixes; consequently, $(\mathsf{com.example.a}\simeq_{1}\mathsf{com.example.b})$ holds; however, $(\mathsf{com.example.a}\leq\mathsf{com.example.b})$ does not hold as $\mathsf{com.example.a}\not\in prec(\mathsf{com.example.b})$.
\end{example}

The next lemma follows immediately from our definitions.

\begin{lemma}\label{lem:notation-functorial}
	The sets $\scrsf{domain}\subsub{\monoid},\scrsf{record}\subsub{\monoid},\scrsf{zone}\subsub{\monoid}, \scrsf{zfc}\subsub{\monoid}, \scrsf{query}\subsub{\monoid}$, $\scrsf{answer}\subsub{\monoid}$, and $\scrsf{slist}\subsub{\monoid}$ are all functorial\footnote{Given an expression $A$ that contains $B$, we say that $A$ is \emph{functorial} in $B$, if $B \mapsto A$ is a functor.} in the monoid $\monoid$.
\end{lemma}

Notably, IP addresses remain unchanged with the change in the underlying monoid $\monoid$.

\subsection{Recursive Communicating Processes with Timers\label{subsec:trcps}}
We will now introduce the notion of recursive communicating processes with timers, extending the model of recursive communicating processes (rCPS) \cite{heussner2010reachability,heussner2012reachability}, where finitely many processes communicate with each other over perfect channels, dictated by an underlying topology. Therefore, we first define the notion of a network topology, which is a directed graph whose vertices represent processes and whose edges represent channels.
\begin{definition}
	A \emph{network topology} $\cal T$ is a 3-tuple $\tup{P,\scrsf{ch},\ch}$, where $P$ is a finite set of processes, $\scrsf{ch}$ is a finite set of perfect channel identifiers, and $\ch : \scrsf{ch} \rightharpoonup P \times P$ is a partial map from channel identifiers to the processes it connects.
	
	We write $\ch_{c}$ to denote the channel from the \emph{source} process $i \in P$ to the \emph{destination} process $j \in P$, if $\ch(c) = \tup{i,j}$. Moreover, we assume that $\ch^{-1}(\tup{i,i})$ is always empty i.e., there is no channel from a process to itself.
\end{definition}

Consider a network topology $\mathsf T = \tup{P,\scrsf{ch},\ch}$. Given a message alphabet $M$, we define the set of \emph{channel actions}, $\scrsf{caction}_{\mathsf T,M}^p = \{\ch_c!a \mid c \in \ch^{-1}\tup{p,p'}, p'\in P, a \in M\} \cup \{\ch_c?a \mid c \in \ch^{-1}\tup{p',p}, p'\in P, a \in M\}$ for every $p \in P$. The operations of the form $\ch_c?a$ and $\ch_{c}!a$ denote the actions of \emph{receiving} and \emph{sending} the message $a$ on channel $\ch_c$ respectively. 

Next, we define the model of pushdown systems with timers and a timeless stack, where we follow the model proposed in \cite[\S 2]{clemente2015timed}, but use timers \cite{vaandrager2023learning,KuroseRoss16} instead of clocks.\footnote{The restriction to timers reduces the space complexity of the abstract configuration space we introduce in Section~\ref{subsec:abstract-model}. From a decidability perspective, all the reductions presented in this paper hold for the more expressive class of timed pushdown systems proposed in \cite{clemente2015timed}.} Unlike clocks, timers run \emph{backwards}, starting from an initial value. Moreover, we restrict the timing constraints to \new{only allow \emph{timeout checks}}. A \emph{timing constraint} over a set of timers $X$ is a formula generated by the grammar,
\[ \scrsf{cc}\subsub{X} \coloneqq x > 0 \mid \scrsf{cc}\subsub{X} \wedge \scrsf{cc}\subsub{X},\]
where $x\in X$ and timer resets are represented by $[\mathbf{x} \leftarrow \mathbf{k}]$, with $\mathbf{k} \in \mathbb{N}^n$; $\mathbf{x}\in X^n$ is a vector with distinct timer components, denoting that each timer $\mathbf{x}_i$ in $\mathbf{x}$ is reset to the value $\mathbf{k}_i$.

A \emph{timer valuation} over $X$ is a map $\nu : X \rightarrow \rat^{\geq 0} \cup \{\infty\}$, where $X$ is a set of timers and each timer either holds a non-negative rational number (denoted by $\rat^{\geq 0}$) or it has not yet been initialized (denoted by $\infty$). Given $k \in \rat^{\geq 0}$, we denote by $\nu - k$ the valuation, where for each $x\in X$, 
\[(\nu - k)(x) = \begin{cases}
	\nu(x) - k, & \text{if } \nu(x) \geq k\\
	0, & \nu(x) < k\\
	\infty, & \nu(x) = \infty.
\end{cases}\]
Given a vector $k \in \nat^n$, we denote by $\nu[\mathbf{x} \leftarrow \mathbf{k}]$ the updated timer valuation such that $\nu[\mathbf{x} \leftarrow \mathbf{k}](x) = k_i$ if $x = \mathbf{x}_i$ for some $i$; otherwise, $\nu[\mathbf{x} \leftarrow \mathbf{k}](x) = \nu(x)$.\footnote{Note that the semantics of a timer is less expressive than that of a clock. However, since most implementations of $\scrsf{dns}$ use the notion of timer-based timeouts in their states, using timers greatly simplifies further reductions in this paper.}

\begin{definition}\label{def:tpds}
	A \emph{pushdown system with timers} (tPDS) is a 6- tuple $\cal P = \tup{S,S_{init},A,\Gamma,X,\trans}$, where $S$ is a finite set of states, $S_{init} \subseteq S$ is a finite set of initial states, $A$ is a finite input alphabet, $\Gamma$ is a finite stack alphabet, $X$ is a finite set of timers, and $\trans$ is a set of transition rules. The transition rules are of the form $s \xrightarrow[{[\mathbf{x} \leftarrow \mathbf{k}]},cc]{\tup{\alpha,a,\alpha'}} s'$, with $s,s' \in S$ states, $a \in A \cup \{\epsilon\}$ an input letter, and $\mathbf{x} \in X^n$ an n-dimensional vector of timers, being reset to $\mathbf{k}$, for some $n$; $cc \in \scrsf{cc}_X$ is a timer constraint; $\alpha\in \Gamma^*$ represents the finite lookup of the top of the stack, and $\alpha'\in \Gamma\cup \varepsilon$ represents the element to be pushed (or none thereof, if $\alpha' = \varepsilon$). The \emph{size} of $\cal P$ is the cumulative size of all its components.
\end{definition}

A \emph{configuration} $C$ of the tPDS is a 3-tuple $\tup{s,\nu,\gamma}$, where $s\in S$, $\nu$ is a \emph{timer valuation} over $X$, and $\gamma \in \Gamma^*\bottom$ is the stack valuation with the stack top being the leftmost label in $\gamma$, and the special symbol $\dashv\,\notin \Gamma$ denoting the bottom of the stack. The set of all configurations is the set $\config{\cal P} = S\times (\rat^{\geq 0} \cup \{\infty\})^X \times \Gamma^*\bottom$. The set of initial configurations is $\config{\cal P}_{init} = S_{init}\times \{\infty\}^{X} \times\{\dashv\}$. A configuration in $\{s\} \times (\rat^{\geq 0} \cup \{\infty\})^Xg \times \Gamma^*\bottom$, for $s \in S$, is called an $s$\emph{-configuration}.

The semantics for the tPDS is the LTS $\tup{\cal P} = \tup{\config{\cal P},\config{\cal P}_{init},A',\trans'}$, where $A' = ((\Gamma^*\cup\Gamma^+\bottom) \times A \times (\Gamma \uplus \{\varepsilon\})) \times \{[\mathbf{x} \leftarrow \mathbf{k}] \mid \mathbf{x} \in X^n, \mathbf{k} \in \nat^n, n \in \nat \} \times \scrsf{cc}_X$ and $\trans'$ contains two kinds of transitions:
\begin{enumerate}[(i)]
	\item \emph{timed transitions} labelled by $t \in \rat^{\geq 0}$, such that $\tup{s, \nu, \gamma} \xrightarrow{t} \tup{s, \nu - t, \gamma}$, and
	\item \emph{discrete transitions}, such that $\tup{s, \nu, \gamma} \xrightarrow[{[\mathbf{x} \leftarrow \mathbf{k}]},cc]{\tup{\alpha,a,\alpha'}} \tup{s',\nu',\gamma'}\in \trans'$, whenever $\nu \models cc$, $\nu' = \nu[\mathbf{x} \leftarrow \mathbf{k}]$, $s \xrightarrow[{[\mathbf{x} \leftarrow \mathbf{k}]},cc]{\tup{\alpha,a,\alpha'}} s' \in \trans$, and we have, 
	\[\gamma' = \begin{cases}
		\alpha'\gamma, & \text{if } \alpha' \neq \varepsilon,\alpha = \epsilon\\
		\gamma, & \text{if } \alpha' = \varepsilon,\alpha = \epsilon\\
		\gamma[1\ldots|\gamma|], & \text{if } \alpha' = \varepsilon,\alpha \neq \epsilon \text{, and } \gamma = \alpha\gamma'' \text{, for some } \gamma'\in \Gamma^*\bottom\\
		\alpha'\cdot\gamma[1\ldots|\gamma|], & \text{if } \alpha' \neq \varepsilon,\alpha \neq \epsilon \text{, and } \gamma = \alpha\gamma'' \text{, for some } \gamma'\in \Gamma^*\bottom
	\end{cases}\]
	Intuitively, each stack operation checks whether a finite length of the top of the stack matches $\alpha$, and performs either a push operation, \emph{nop} (i.e., let the stack remain the same), pop operation, or a \emph{pop and push} operation. 
\end{enumerate}
The runs of $\tup{\cal P}$ are slightly different from the definition for LTS given in \Cref{subsec:lts}. A \emph{run} in $\tup{\cal P}$ is a finite \emph{alternating} sequence of the form $\rho = C_0 \xrightarrow[{[\mathbf{x}_1 \leftarrow \mathbf{k}_1]}, cc_1]{\tup{\alpha_1,a_1,\alpha'_1}} C_1 \xrightarrow{t_1} C_1' \xrightarrow[{[\mathbf{x}_2 \leftarrow \mathbf{k}_2]}, cc_2]{\tup{\alpha_2,a_2,\alpha'_2}} C_2 \ldots C_{n-1}' \xrightarrow[{[\mathbf{x}_n \leftarrow \mathbf{k}_n]}, cc_n]{\tup{\alpha_n,a_n,\alpha'_n}} C_n$, where $t_j \in \rat^{\geq 0}$ and for every $1 \leq i \leq n$, $C_{i-1}' \xrightarrow[{[\mathbf{x}_n \leftarrow \mathbf{k}_n]}, cc_i]{\tup{\alpha_i,a_i,\alpha'_i}} C_i \in \scrsf{trans}'$. The trace of $\rho$ is the sequence $a_1\ldots a_n \in A^*$. If the set of timers $X$ is empty, $\cal P$ behaves as a standard \emph{untimed} pushdown system (PDS).

\subsubsection{Recursive Communicating Processes with Timers}
We follow the definition of an rCPS proposed in \cite[Definition 1.7]{heussner2012reachability}, extending it with timers.

A \emph{system of recursive communicating processes with timers} (trCPS) $\cal R = \tup{\mathsf T,\mathsf M, (\cal P^p)_{p \in P}}$, is given by a network topology $\mathsf T = \tup{P, \scrsf{ch}, \ch}$, a \emph{message} alphabet $\mathsf M$, and for each process $p \in P$, a tPDS $\cal P^p = \tup{S^p,S_{init}^p,A^p,\Gamma^p,X^p,\trans^p}$ such that $\scrsf{caction}^p_{\cal T, M} \subseteq A^p$. We write $\scrsf{s} = \prod_{p \in P} S^p$ for the set of \emph{global states}. The topology $\mathsf T$ is \emph{pointed} with respect to $\cal R$ if exactly one process in $P$ corresponds to a tPDS and all the others correspond to an LTS with finitely many states (see \Cref{fig:pointed-topo}).

\begin{wrapfigure}{r}{0.35\textwidth}
	\centering
  \resizebox{.24\textwidth}{!}{
	\begin{tikzpicture}%
		\node (r) [fill=red!10,draw=red,circle] at (0,0) {};
		\node (p1) [fill=blue!10,draw=blue, circle] at (0,-2) {};
		\node (p2) [fill=blue!10,draw=blue, circle] at (-1.8,-1) {};
		\node (p3) [fill=blue!10,draw=blue, circle] at (-1,-1.8) {};
		\node (p4) [fill=blue!10,draw=blue, circle] at (1.8,-1) {};
		\node (p5) [fill=blue!10,draw=blue, circle] at (1,-1.8) {};
		
		\path[->] (r) edge [bend right=20]             node              {} (p1)
		edge [bend right=20]             node              {} (p2)
		edge [bend right=20]             node              {} (p4)
		(p1) edge [bend right=20]             node              {} (r)
		edge             node              {} (p3)
		edge             node              {} (p5)
		(p2) edge [bend right=20]             node              {} (r)
		(p3) edge node {} (p2)
		edge node {} (p4)
		(p4) edge [bend right=20]             node              {} (r);
	\end{tikzpicture}
 }
	\captionsetup{skip=5pt}
	\caption{\label{fig:pointed-topo}A pointed topology. The red  process is the only tPDS in the network.}
\end{wrapfigure}

\sloppy The semantics of trCPS is defined in terms of a global LTS $\tup{\cal R} = \tup{\scrsf{config}, \scrsf{config}_{init}, \mathsf{A}, \trans}$, where $\scrsf{config} = (\config{\cal P^p})_{p \in P} \times (\mathsf M^*)^{\scrsf{ch}}$ is the set of \emph{global configurations}, $\scrsf{config}_{init} = (\config{\cal P^p}_{init})_{p \in P} \times \{\epsilon\}^{\scrsf{ch}}$ is the set of initial configurations, and $\mathsf{A} = \biguplus_{p \in P} {A'^p}$ is the set of actions. Given a \emph{global configuration} $\scrsf{c} \in \scrsf{config}$, we write $\scrsf{c}^p \in \config{\cal P^p}$ to denote the \emph{local configuration} in $\scrsf{c}$ that corresponds to the process $p \in P$, and $\scrsf{c}^{c}$ to denote the contents of channel $\ch_c$ such that $\scrsf{c}^c = \scrsf{c}^{c'}$ if $\ch_c = \ch_{c'}$. A configuration in $\scrsf{c}$ is called an $\mathbf{s}$\emph{-configuration} when, for every $p \in P$, $\scrsf{c}^p$ is an $s^p$-configuration, for some $\mathbf{s} =\tup{s^p}_{p\in P} \in \scrsf{s}$.

The discrete transitions in $\trans$ of the form $\scrsf{c} \xrightarrow[{[\mathbf{x} \leftarrow \mathbf{k}]}, cc]{\tup{\alpha,a,\alpha'}} \scrsf{c}'$, for $a\in A^p$, satisfy the following:
\begin{enumerate}[(i)]
	\item $\scrsf{c}^p \xrightarrow[{[\mathbf{x} \leftarrow \mathbf{k}]}, cc]{\tup{\alpha,a,\alpha'}} \scrsf{c}'^p \in \trans'^p$ and $\scrsf{c}^q = \scrsf{c}'^q$ for all $q\in P \setminus p$,
	\item if $a \in A^p\setminus\scrsf{caction}^p_{\cal T, M}$, $\scrsf{c}^c = \scrsf{c}'^c$ for $c \in \scrsf{ch}$,
	\item if $a = \ch_c!m$ then $\scrsf{c}'^{c''} = \scrsf{c}^{c''}\cdot m$ for each $c'' \in \ch^{-1}(\ch_c)$ and $\scrsf{c}'^{c''} = \scrsf{c}^{c''}$ for $c'' \in \scrsf{ch}\setminus\ch^{-1}(\ch_c)$,
	\item if $a = \ch_c?m$ then $m\cdot\scrsf{c}'^{c''} = \scrsf{c}^{c''}$ for each $c'' \in \ch^{-1}(\ch_c)$ and $\scrsf{c}'^{c''} = \scrsf{c}^{c''}$ for $c'' \in \scrsf{ch}\setminus\ch^{-1}(\ch_c)$.
\end{enumerate}
The timed transitions in $\trans'$ of the form $\scrsf{c} \xrightarrow{t} \scrsf{c'}$ satisfy $\scrsf{c}^p \xrightarrow{t} \scrsf{c}'^p \in {\trans^p}'$ for all $p \in P$. We will refer to runs in $\tup{\cal R}$ as runs in $\cal R$.

Observe that, if none of the component processes in a trCPS uses timing operations, our model essentially becomes a rCPS \cite[Def. 1.7]{heussner2012reachability}. Similarly, without timing and stack operations, the model becomes a CPS \cite[Def. 1.2]{heussner2012reachability}. 

\subsubsection{Eager Runs}\label{subsubsec:eagerness}
The notion of eager runs was first introduced by Heussner et. al. \cite{heussner2010reachability,heussner2012reachability}, in the context of rCPS. We will extend this notion to trCPS.

\begin{definition}
	A run $\rho = \scrsf{c}_0 \xrightarrow[{[\mathbf{x}_1 \leftarrow \mathbf{k}_1]}, cc_1]{\tup{\alpha_1,a_1,\alpha'_1}} \scrsf{c}_1 \xrightarrow{t_1} \scrsf{c}_1' \xrightarrow[{[\mathbf{x}_2 \leftarrow \mathbf{k}_2]}, cc_2]{\tup{\alpha_2,a_2,\alpha'_2}} \scrsf{c}_2 \ldots \scrsf{c}_{n-1}' \xrightarrow[{[\mathbf{x}_n \leftarrow \mathbf{k}_n]}, cc_n]{\tup{\alpha_n,a_n,\alpha'_n}} \scrsf{c}_n$ in a trCPS is \emph{eager} if for all $1 \leq i \leq n$, if $a_i$ is of the form $\ch_c?m$ then $i > 1$ and $a_{i - 1}$ is the corresponding write action $\ch_c!m$ in $\rho$. 
\end{definition}

Intuitively, \emph{eagerness} formalizes that, in a run, send actions are immediately proceeded by their matching receive actions. A configuration $\scrsf{c} \in \scrsf{config}$ is \emph{eager-reachable} in a trCPS $\cal R$ if there exists an eager run from an initial configuration in $\scrsf{config}_{\mathit{init}}$ to $\scrsf{c}$. We say that a trCPS $\cal R$ is \emph{eager} when every run in $\cal R$ is eager.
The \emph{state eager-reachability problem} for $\cal R$ asks whether, for a given global state $\mathbf{s} \in \scrsf{s}$, there exists an eager run from an initial configuration to an $\mathbf{s}$-configuration.

Heussner et. al. \cite[Prop. 2.4]{heussner2010reachability} characterized a subclass of topologies for which the state eager-reachability problem for rCPS, satisfying certain additional restrictions on their communication behaviour, is decidable. As it happens, the class of rCPS over pointed topologies satisfies these conditions; in particular, this fact is subsumed by the case $(ii)$ in their proof, where they reduce the search for an eager run in the rCPS to the search for a run in a (untimed) pushdown system. The reduced pushdown system even preserves the set of traces of all eager runs in the original rCPS. In the case of pointed topologies, such a reduction leads to a pushdown system of size polynomial in the size of each of the processes.

\begin{lemma}\label{lem:rcps-to-pds}
	Given any rCPS $\cal R$, over a pointed topology, it is $\mathsf{PolyTime}$-reducible to an untimed pushdown system $\cal P$ with the set of states as the set of global states in $\cal R$, such that, for every eager run in $\cal R$, from an $\mathbf{s}$-configuration to an $\mathbf{s}'$-configuration, there exists a run in $\cal P$, from an $\mathbf{s}$-configuration to an $\mathbf{s}'$-configuration, with the same trace.
\end{lemma}

\subsection{DNS Semantics}\label{subsec:dns-semantics}

Given an arbitrary zone file configuration $\tup{\cal Z ,\Theta} \in \scrsf{zfc}\subsub{\monoid}$, we model the corresponding instance of DNS as the trCPS $\scrsf{dns}\subsub{\cal Z}=\tup{\mathsf T\subsub{\cal Z},\mathsf M\subsub{\cal Z}, (\cal P^p)_{p\in P\subsub{\cal Z}}}$, where $\mathsf T\subsub{\cal Z} = \tup{P\subsub{\cal Z}, \scrsf{ch}\subsub{\cal Z}, \ch\subsub{\cal Z}}$, and $\mathsf M\subsub{\cal Z} = \scrsf{query}\subsub{\monoid}\uplus\scrsf{answer}\subsub{\monoid}$ is the disjoint union of the set of all queries and answers, such that the following hold:
\begin{enumerate}[(i)]
	\item $P\subsub{\cal Z} = \{ r \} \uplus \cal S \subseteq \scrsf{ip}$,
	\item $\scrsf{ch}\subsub{\cal Z} = \cal S \times \cal S \uplus \{ r \} \times \cal S \uplus \cal S \times \{ r \}$,
	\item $\ch\subsub{\cal Z}$ is the identity map,
	\item $\cal P^p$ is a tPDS of the form $\rec{p}$, if $p = r$, and an LTS of the form $\ns{p}$, otherwise.
\end{enumerate}

Intuitively, $\rec{r}$ denotes the resolver, identified by its IP address; there are no communication channels between the resolver to itself. This follows, by convention \cite{groot,dnsmaude}, from a strict restriction on our model, where we typically do not allow a nameserver to also serve as a resolver in the network. However, note that this is a logical separation, as in reality this is not the case; resolvers may indeed themselves behave as nameservers \cite{rfc1034,rfc1035} in a network; $\ns{p}$ is the formal model of a nameserver. We define these models below, in Sections \ref{subsec:NS-semantics} and \ref{subsec:RR-semantics}. Note that $\mathsf T_{\cal Z}$ is a pointed topology with only one tPDS process $\cal P^r$. \new{We shall crucially use this fact in~\Cref{sec:property-validation-to-reachability} to obtain the decidability of the $\scrsf{dnsverif}$ problem.}

Unless otherwise stated, we will implicitly assume that the underlying monoid $\monoid$ is the free monoid $\cal L$, as described in Example \ref{ex:normal-domain}. For notational convenience, we shall thus drop the corresponding subscript $\cal L$ from all our terms.  
For the rest of the paper, we also arbitrarily fix an instance of DNS with the zone file configuration $\tup{\cal Z : \cal S \rightarrow \cal P(\scrsf{zone}),\Theta}$, denoting it as $\scrsf{dns}$, by dropping the subscript $\cal Z$. If clear from the context, we shall also use $\cal Z$ interchangeably to denote both, the set of all zones in the image of the configuration, and the set of all records in those zones.

\subsection{\label{subsec:NS-semantics}Authoritative Nameserver}

An \emph{authoritative nameserver}, $\ns{s}$, identified by its IP address, $s \in \cal S$, is an LTS of the form $\tup{\scrsf{nsconfig}_s,\{\tup{\mathit{init},\bot,\bot}\},\scrsf{saction}_s,\trans_s}$, where $\scrsf{nsconfig}_s$ is the set of tuples $\{\mathit{init},\mathit{lookup},\mathit{exact},\mathit{wildcard}\}\times \scrsf{query} \times \cal R$, $\scrsf{saction}_s = \scrsf{caction}^s_{\mathsf T,M} \cup \{in\}$ with the special action label $in$ for internal actions, and $\trans_s \subseteq \scrsf{nsconfig}_s \times \scrsf{saction} \times \scrsf{nsconfig}_s$. 

Given a configuration $p = \tup{\mathit{status}, q, r}$,
\begin{enumerate}
	\item \emph{status} is $\mathit{init}$ when the nameserver does not have pending query requests, $\mathit{lookup}$ when it must start the lookup process for $q$, $\mathit{exact}$ when $\text{dn}(q)$ finds an exact match with the domain name field of some resource record in a zone that $\ns{s}$ is authoritative for, and $\mathit{wildcard}$ when $\text{dn}(q)$ finds a wildcard match with the domain name field of some resource record in such a zone.
	\item $q$ is the current query.
	\item $r$ is the identity of the resolver that the nameserver communicates with.
\end{enumerate}
\begin{figure}
	\resizebox{.75\textwidth}{!}{\begin{tcolorbox}[width=.9\textwidth]
			\begin{minipage}{\textwidth}
			\begin{align*}
		\mathcal{N}(s,q) & =max_{\text{dn}}\{z\in\mathcal{Z}(s)\mid\text{dn}(z)\leq\text{dn}(q)\} \quad
		\mathcal{M}(s,q)  =\bigcup_{z\in\mathcal{N}(s,q)}\{rr\in max_{<_{q,z}}z\}\\
		{\cal T}(s,q) & =\{rr\in{\cal M}(s,q)\mid\text{ty}(rr)=\text{ty}(q)\}\quad\quad\quad
		\tau(s,q)  =\{\text{ty}(rr)\mid rr\in{\cal M}(s,q)\}\\
		{\cal T}_{t}(s,q) & =\{rr\in{\cal M}(s,q)\mid\text{ty}(rr)=t\}\quad\quad\quad\quad\quad
		{\cal G}(s,q)  ={\cal T}_{\scrsf{ns}}(s,q)\cup{\cal G}_{A}(s,q)\\ \\
		{\cal G}_{A}(s,q) & =\{rr\in\bigcup{\cal Z}(s)\mid\exists rr'\in{\cal T}_{\scrsf{ns}}(s,q).\text{val}(rr')=\text{dn}(rr) \wedge\text{ty}(rr)=\scrsf{A}\}
	\end{align*}
	\end{minipage}
	\end{tcolorbox}}
	\captionsetup{skip=5pt}
	\caption{\label{fig:aux-op}Some auxiliary operations. }
\end{figure}

Additionally, we define auxiliary operations in Figure \ref{fig:aux-op}. They
all take the server identity $s$ and an input query
$q$ as input. ${\cal N}(s,q)$ obtains the zones whose domain names have
a maximal match with $q$. ${\cal M}(s,q)$ accumulates the resource
records with maximal $\scrsf{rank}$ in the zones in ${\cal N}(s,q)$. ${\cal T}_{t}(s,q)$ gathers the resource records in ${\cal M}(s,q)$
of type $t$. ${\cal T}(s,q)$ is a special case of ${\cal T}_{t}(s,q)$
when $t=\text{ty}(q)$. $\tau(s,q)$ gathers the types of the
resource records in ${\cal M}(s,q)$. ${\cal G}(s,q)$ yields
the $\scrsf{NS}$ records in ${\cal M}(s,q)$ that have an associated $\scrsf{A}$ record. Finally, ${\cal G}_{A}(s,q)$ denotes the set of associated $\scrsf{A}$ records. Since the operators are parameterized only by a DNS query and the
choice of nameserver, they can be computed from the current state
information in $\mathsf{NS}_{s}$.

\subsubsection{Transitions}

The transitions model the local steps taken by the nameserver during query resolution. In \Cref{fig:ns-transitions}, we illustrate a transition in this set.

\begin{figure*}[h!]
	\centering
	\resizebox{\textwidth}{!}{\begin{tikzpicture}[every label/.style={circle,draw,minimum size=9pt,inner sep=0.5pt,font=\footnotesize}]
			\node [fill=blue!10,draw,rounded corners=5pt] {$\dfrac{\text{dn}({\cal M}(s,q))<\text{dn}(q),\scrsf{dname}\in\tau(s,q), dn\in\scrsf{domain}, {\cal T}_{\scrsf{dname}}(s,q)=\{rr\}, rr'=\tup{\text{dn}(q),\scrsf{cname},\text{ttl}(rr),\text{val}(rr)\text{dn}(rr)^{-1} \text{dn}(q),1}}{\tup{\mathit{lookup},q,r}\xrightarrow{c_{(dn,s),r}!\tup{\scrsf{ansq},\{rr,rr'\}}}\tup{\mathit{init},\bot,\bot}}$};
	\end{tikzpicture}}
	\captionsetup{skip=5pt}
	\caption{\label{fig:ns-transitions}A transition in $\ns{s}$. The premise of the transition specifies the condition when the transition is taken. The action label is above the arrow.}
\end{figure*}

In this transition, the nameserver $\ns{s}$ is about to start the lookup process, initiated by the status $\mathit{lookup}$, for the current query $q$, requested by the resolver $r$. The premise component $\text{dn} (\cal M (s, q)) < \text{dn} (q)$ states that $\text{dn} (q)$ is an extension of the domain name of every resource record in $\cal M (s, q)$, where $\cal M (s, q)$ represents the set of the resource records that the nameserver may potentially send back in the response to the query $q$; the components $\scrsf{dname}\in \tau (s, q)$ and ${\cal T}_{\scrsf{dname}} (s, q) = {rr}$ state that the record $rr$ in $\cal M (s, q)$ is a $\scrsf{dname}$ record. In its response, the nameserver must send back the record $rr$ and its $\scrsf{cname}$ derivative synthesized for $q$ in order to prompt a rewrite of the query. The premise component $rr'=\tup{\text{dn}(q),\scrsf{cname},\text{ttl}(rr),\text{val}(rr)\text{dn}(rr)^{-1} \text{dn}(q),1}$ represents the synthesized record as denoted by the fact that $\text{syn} (rr') = 1$, and $\text{val}(rr')$ is the rewritten domain name. With this, the nameserver then sends the answer $\tup{\scrsf{ansq},\{rr,rr'\}}$, which prompts a rewrite ($\scrsf{ansq}$) using the records $rr$ and $rr'$, and moves back to its initial configuration.

\begin{wrapfigure}{r}{0.3\textwidth}
	\centering
 \resizebox{.3\textwidth}{!}{
	\begin{tikzpicture}[scale=0.5,every label/.style={circle,draw,minimum size=9pt,inner sep=0.5pt,font=\footnotesize}]
		\node (r1) [fill=blue!10,draw,rounded corners=5pt] {$\dfrac{\text{dn}(\cal M (s,q))=\text{dn}(q)}{\tup{\mathit{lookup},q,r}\xrightarrow{}\tup{\mathit{exact},q,r}}$};
		\node [below=20pt,fill=blue!10,draw,rounded corners=5pt] {$\dfrac{\text{dn}(\cal M (s,q))\ni_\star\text{dn}(q)}{\tup{\mathit{lookup},q,r}\xrightarrow{}\tup{\mathit{wildcard},q,r}}$};
	\end{tikzpicture}
 }
	\captionsetup{skip=4pt}
	\caption{\label{fig:ns-transitions-more}Two other transitions.}
 \vspace{-2ex}
\end{wrapfigure}

For the most part, the transitions correspond to the DNS lookup semantics defined in \cite[Figure 3]{groot}.  \Cref{fig:ns-transitions-more}  demonstrates two other nameserver transitions that showcase two tests that are run by the nameserver on the incoming query name. In particular, notice their premises: $\text{dn}(\cal M (s,q)) = \text{dn}(q)$ checks whether the query name can be exactly matched with the name of a resource record in $\cal M (s,q)$. If so, $\ns{s}$ transitions from the $\mathit{lookup}$ status to $\mathit{exact}$. In the second rule,  $\text{dn}(\cal M (s,q)) \ni_\star \text{dn}(q)$ checks whether there is a wildcard record name that pattern-matches with the current query. If so, the system transitions to $\mathit{wildcard}$ status.

\subsection{Resolver}\label{subsec:RR-semantics}

The \emph{resolver} $\rec{r}$ is a tPDS of the form $\tup{\scrsf{resstat}_r, \{\mathit{init}_r\}, \scrsf{raction}_r, \scrsf{query}, X\subsub{\cal Z}, \trans_r}$, where, $\scrsf{raction}_r = \scrsf{caction}^r_{\mathsf T, M} \cup \{in\}$ and $X\subsub{\cal Z} = \{x_{rr}\}_{rr \in \cal Z}$ is the set containing a unique timer for each record in $\cal Z$. Next, we describe the sets $\scrsf{resstat}$ and $\trans_r \subseteq \scrsf{resstat}_r \times \scrsf{raction} \times \scrsf{resstat}_r$, of states and transitions, respectively.

The set $\{\mathit{init},\mathit{query}, \mathit{wait}\} \times \scrsf{slist}$, is the set of control states in $\rec{r}$, and $\mathit{init}_r = \tup{\mathit{init},\epsilon}$ is the initial control state. Intuitively, given a configuration $\tup{\tup{\mathit{status},sl}, \mathit{cache}, st}$ of the tPDS $\rec{r}$:

\begin{enumerate}[(i)]
	\item $\mathit{status}$ is $\mathit{init}$ when the resolver awaits a query request, $\mathit{query}$ when it has an active query to resolve, and $\mathit{wait}$ when the resolver awaits a response for the active query under resolution.
	\item $sl$ maintains a list of domain names, which correspond to nameservers, available to the resolver for query resolution, in a decreasing order of \emph{priority}. This means, if $d = sl[i]$ for some $0 \leq i \leq |sl| - 1$, then the resolver understands that $d$ is the domain name of a nameserver in the network that may potentially be able to answer its query; moreover, it expects that the IP address of the nameserver itself is stored in its cache. The action $\odot$ denotes the addition of a nameserver reference to $sl$, while potentially assigning it the maximum priority, if it is added to the left of $sl$. In practice, the resolver only queries the leftmost domain name in $sl$. If for some reason, the leftmost domain name gets removed, the next domain name becomes available for querying.
	\item $st\in\scrsf{query}^*$ denotes the current stack contents of the configuration, recording the query hierarchy in case of subquery resolution. The stack top denotes the present subquery under resolution and the query at the stack bottom
	denotes the original query to be resolved.
	\item $cache$ denotes the resolver's current cache contents. If the valuation for the timer $x_{rr}$, associated with the record $rr$, satisfies the constraint $x_{rr} > 0$, i.e., $cache(x_{rr}) > 0$, it means that the record $rr$ is still alive in the cache of the resolver and can be used for query resolution. Equivalently, we will also write $rr \in cache$ to denote this fact.
\end{enumerate}

\subsubsection{Transitions}

In \Cref{fig:rec-transitions}, we illustrate some transitions in $\trans_r$. The action sequence labels are above the arrow and the clock constraints below. To simplify the presentation, we have concatenated the action labels to denote the execution of a sequence of actions. Below, we describe  these transitions in more detail.

\paragraph{Transition \circled{1}} %

Before this transition, the resolver status is $\mathit{query}$, indicating that it is about to resolve the query $q_1$ at the top of the stack. The premise component $\text{dn}(q_1) \leftarrow_{\scrsf{a}} rr$, for $rr\in cache$, states that $q_1$ can be resolved with the $\scrsf{a}$ record $rr$ in the resolver's cache; this means that the resolver can immediately resolve $q_1$ and thus pop the top of the stack. Note that the resolver checks that the record $rr$ has not yet expired, with the clock constraint $x_{rr} > 0$. Since the action component $q_1\cdot q_2$ checks that the rest of the stack is non-empty, the resolver can then revert to the resolution of the next stack entry (which is the previous subquery) $q_2$. It does this by sending the query $q_2$ to the nameserver with the IP address $\text{val}(rr)$, which is assumed to have the domain name $\text{dn}(q_1)$ (using the premise). After the transition, the resolver moves to the status $\mathit{wait}$ as it waits for a response from the network, reflected by the state $\tup{\mathit{wait},sl}$.

Intuitively, this transition is taken when the resolver is about to successfully resolve the latest \emph{subquery} at the top of the stack required to resolve the penultimate (sub)query and seemingly moves one step closer to the resolution of the original query (i.e., reduces the stack depth by one).

\begin{figure}
	\centering
  \begin{minipage}{.45\linewidth}
		\label{subfig:trans-rec-a}
		\resizebox{.85\linewidth}{!}{
  \begin{tikzpicture}[every label/.style={circle,draw,minimum size=9pt,inner sep=0.5pt,font=\footnotesize}]
				\node [fill=blue!10,draw,rounded corners=5pt, label=north west:1] {$\dfrac{\text{dn}(q_{1})\leftarrow_{\scrsf{a}}rr}{\tup{\mathit{query},sl}\xrightarrow[x_{rr} > 0]{\tup{q_1q_2,\ch_{r,\text{val}(rr)}!q_{2},q_2}}\tup{\mathit{wait},sl}}$};
				\end{tikzpicture}
    }
    \end{minipage}
    \begin{minipage}{.52\linewidth}
       		\label{subfig:trans-rec-c}
		\resizebox{.9\linewidth}{!}{
  \begin{tikzpicture}[every label/.style={circle,draw,minimum size=9pt,inner sep=0.5pt,font=\footnotesize}]
			\node [fill=blue!10,draw,rounded corners=5pt, label=north west:3] {$\dfrac{s\in \cal S, list \in Lin(Val(R))}{\tup{\mathit{wait},sl}\xrightarrow[{[(x_{rr})_{rr\in R}\leftarrow (\text{ttl}(rr))_{rr\in R}]}]{\tup{q, \ch_{s,r}?\tup{\scrsf{ref},R}, q}}\tup{\mathit{query},list \odot sl}}$};
			\end{tikzpicture}
   }
   \end{minipage}
    \begin{minipage}{.6\linewidth}
    \label{subfig:trans-rec-b}
		\resizebox{\linewidth}{!}{
  		\begin{tikzpicture}[every label/.style={circle,draw,minimum size=9pt,inner sep=0.5pt,font=\footnotesize}]
				\node [fill=blue!10,draw,rounded corners=5pt, label=north west:2] {$\dfrac{ \neg(ns\leftarrow_{\scrsf{a}}rr \vee \text{dn}(q)\leftarrow_{\{\text{ty}(q),\scrsf{dname}\}} rr) \wedge ns = sl[min_{sl[i] < \text{dn}(q)} i]}{\tup{\mathit{query},sl}\xrightarrow[x_{rr} > 0]{\tup{q,in,\tup{ns,\scrsf{a}}\cdot q}}\tup{\mathit{query},sl}}$};
		\end{tikzpicture}
    }
    \end{minipage}
  \begin{minipage}{.6\linewidth}
	\scriptsize
	\begin{align*}
	\leftarrow_{T} & \coloneqq\{(dn,rr)\mid\text{dn}(rr)\leq dn,\text{ty}(rr)\in T\}\subseteq\scrsf{domain}\times\scrsf{record}\\
	Val(R) & \coloneqq \{\text{val}(rr)\mid rr\in R,\text{ty}(rr)=\scrsf{ns}\} \subseteq \scrsf{domain}\\
	Lin(V) & \coloneqq \{m_1\ldots m_{|V|} \in V^{|V|} \mid m_i \neq m_j, 0 \leq i < j < |V| \}
	\end{align*}
 \end{minipage}
	\captionsetup{skip=8pt}
	\caption{\label{fig:rec-transitions}\new{Some transitions in $\rec{r}$.}}
\end{figure}

\paragraph{Transition \circled{2}}
Prior to this transition, the resolver configuration is similar to the previous one. However, this time the premise implies that the query (or subquery) $q$ cannot be resolved or rewritten using the cache.\new{ To resolve any further, the resolver picks the leftmost domain $ns = sl[min_{sl[i] < \text{dn}(q)} i]$ in $sl$ which is a prefix of the current query name.} The next step is to derive the IP for $ns$ from the cache. Consequently, the resolver pushes $(ns,\scrsf{a})$ onto the query stack in order to find the IP address of the nameserver with the domain name $ns$.

\paragraph{Transition \circled{3}}
In this transition, the resolver is waiting for a response from a nameserver and receives a referral response accompanied by a set of resource records $R$. The resolver then adds the records in $R$ to its cache, instantiating a timer for each element and setting the timer to the corresponding time-to-live. Additionally, the resolver filters out the $\scrsf{ns}$ records from $R$, represented by $\mathit{Val}(R)$, and for each such record, say $rr$, makes the server list entry $\text{val}(rr)$, dictated by the priority order in $list$; $\mathit{Lin}(\mathit{Val}(R))$ computes the set of all possible priorities. Since the resolution of the original query is not yet complete, the resolver updates its status to $\mathit{query}$.

\new{Note that each of the above transitions only checks for the existence of a record in the resolver's cache. Hence, one can force the transition to execute without any elapse of time, effectively making the cache elements seemingly persistent. 
However, DNS behaves statefully upon the \emph{absence} of records as well. For the sake of clarity and ease of presentation, the set of all resolver transitions has been moved to our technical report~\cite[Appendix A.2]{tech-rpt}.
}

\begin{rem}
	Real-world nameservers are often caching nameservers \cite{rfc1034,rfc1035}, i.e., they are equipped with limited resolution capabilities (they are unable to perform subquery resolution) and caching capabilities (to store answers from resolved queries). Resolvers are sometimes configured to cache negative answers with the $\scrsf{nxdomain}$ label. In this paper, we do not assume such systems to simplify our presentation. However, this is not a restriction of our model as our results hold even otherwise.
\end{rem}

\subsection{DNS is Eager}

We conclude this section by showing that the formal model $\scrsf{dns}$ is eager.

\begin{lemma}\label{lem:eager-dns}
	\scrsf{dns} is eager.
\end{lemma}
\begin{proof}[Proof Sketch.]
    The proof consists of two steps.
    First,
    we show that every reachable configuration in $\scrsf{dns}$ contains exactly one process capable of either sending or receiving a message. 
    Second,
    we show that once a process sends (or receives) a message, it transitions to a state where it can only receive (or send) messages. 
    Now, since the runs in $\scrsf{dns}$ are interleavings of local runs of processes, every transition is executed by \emph{exactly} one process. Hence, any send action must be immediately followed by a receive action.

    The first step follows by induction on the length of the runs. Observe that the initial configuration only allows the resolver process to send a message, thus covering the base case. The second step can be shown with a case-by-case analysis of the transitions (see \cite[Appendix A]{tech-rpt} for the full set of transitions).
\end{proof}

Hence, considering an arbitrary run in $\scrsf{dns}$ is equivalent to considering an eager run. We will use the above lemma, together with Lemma \ref{lem:rcps-to-pds}, to reduce the $\verif$ problem to the LTS verification problem for a pushdown system. In Section \ref{sec:property-agnostic} we prove that $\scrsf{dns}$ can in fact be reduced to an rCPS with an \emph{equivalent} set of traces.

	\section{\label{sec:property-agnostic}Property-Agnostic Abstractions}

We have so far established that $\scrsf{dns}$ can be modelled as a trCPS. The set of its global control states $(\scrsf{nsconfig}_{s})_{s\in\cal S} \times (\scrsf{resstat}_r)_{r\in\cal R}$, as well as the message alphabet $\mathsf M$, are infinite. A priori, this prevents us from using any reachability results from \Cref{subsubsec:eagerness}. In particular, this is caused by three main factors: (1) an infinite query space, (2) clock valuations, and (3) an unbounded stack.

In this section, we  consider the first two factors and develop ways to bound them, consequently reducing the $\scrsf{dns}$ to an rCPS. Broadly speaking, we will use finite monoids as an abstraction for the query space. To abstract time, we will use a variant of the region abstraction method proposed in \cite{alur1994theory}. In the next section, we will formally define the $\verif$ problem and derive a decision procedure for it.

Let us revisit some basic concepts from algebraic language theory. We first fix a finite alphabet $\Sigma$.

\begin{definition}[Syntactic Congruence]
	Given a regular language $L\subseteq \Sigma^*$, the \emph{syntactic congruence} with respect to $L$ is the coarsest congruence relation $\sim_L$ on $\Sigma^*$ such that,
	\[u \sim_L v \iff \forall x,y\in\Sigma^*. (xuy \in L \Leftrightarrow xvy \in L).\]
\end{definition}
Observe that, if $u_1 \sim_L v_1$ and $u_2 \sim_L v_2$, then $u_1\cdot u_2 \sim_L v_1\cdot v_2$. Consequently, the quotient structure $\tup{L/\sim_L, \cdot}$ forms a finite monoid \cite{nerode1958synmon,rabinscott1959synmon}, called the \emph{syntactic monoid} of $L$. We denote it by $\scrsf{syn}(L)$. This induces a canonical surjective homomorphism from the free monoid $\Sigma^*$ to $\scrsf{syn}(L)$, called the \emph{syntactic homomorphism} with respect to $L$, denoted by $\eta_L$. The syntactic monoid of a language can be effectively constructed \cite{kozen2012automata} as the transition monoid of its minimal automaton \cite{1957fundamental}.

\subsection{Positive Prefix-Testable Languages}\label{subsec:ppt}
A regular language $L$ is \emph{prefix-testable} \cite{almeida1995finite,van2013separation}, if membership of $L$ can be tested by inspecting prefixes up to some length, namely if $L$ is a finite Boolean combination of languages of the form $u\Sigma^*$, for a finite word $u$. 
We will use the \emph{positive} fragment of prefix-testable languages to construct a sound and complete abstraction for the set of domain names $\scrsf{domain}$. A \emph{positive prefix-testable} (PPT) language is a finite combination of languages of the form $u\Sigma^*$, using only the union and intersection operations.\footnote{The class of positive prefix-testable languages is an old player in algebraic language theory, and is a simple example of what is known as a positive variety of regular languages \cite{pin1995variety}. Syntactic congruence for such languages corresponds to the class of ordered monoids. However, unordered monoids suffice for our purposes.}

\begin{example}\label{ex:syntactic-monoid}
	Given $\Sigma = \{a,b\}$, the language $L_0 = ab\Sigma^* \cup b\Sigma^*$ is a PPT language.~\Cref{subfig:ex-minimal-automata} shows the minimal automaton $\cal A_0$ for the language;~\Cref{subfig:adj-matrix} shows the elements in the syntactic monoid $\synmon{L_0}$, i.e., the transition monoid of $\cal A_0$. This can be constructed by first computing the adjacency matrices $M_a$ and $M_b$ that correspond to the $a-$ and $b-$ labelled edges in the automaton's underlying graph. We then close them under matrix multiplication, such that the transition matrix for the word $w$ is denoted by $M_w$. For instance, the multiplication $M_{aa}\cdot M_{ab}$ yields the matrix $M_{aa}$ again. Since $M_{bb} = M_b\cdot M_b = M_b$, it has not been included in the set.
	\begin{figure}
		\centering
		\subfloat[]{
			\centering
			\label{subfig:ex-minimal-automata}
			\resizebox{0.25\textwidth}{!}{\begin{tikzpicture}
					[shorten >=1pt,node distance=2cm,on grid,>={Stealth[round]},auto,initial text=,
					every state/.style={draw=red!50,very thick,fill=red!20}]
					\node[state,initial]  (q_0)                      {$q_0$};
					\node (A) [above=of q_0] {$\cal A_0$};
					\node[state]          (q_1) [right=of q_0] {$q_1$};
					\node[state,accepting]          (q_2) [below right=of q_1] {$q_2$};
					\node[state](q_3) [above right=of q_1] {$q_3$};
					
					\path[->] (q_0) edge              node              {$a$} (q_1)
									           edge [bend right=30]             node [swap] {$b$} (q_2)
					                (q_1)  edge              node              {$a$} (q_3)
					                           edge              node              {$b$} (q_2)
					                (q_2) edge [loop below] node        {$a,b$} ()
					                (q_3) edge [loop above] node {$a,b$} ();
		\end{tikzpicture}}}\hfill
	\subfloat[]{
		\centering
		\label{subfig:adj-matrix}
		\resizebox{0.7\textwidth}{!}{\begin{tikzpicture}
				\node (m_a) at (-2,0) {$M_a = \begin{bmatrix}
																		0 & 1 & 0 & 0\\
																		0 & 0 & 0 & 1\\
																		0 & 0 & 1 & 0\\
																		0 & 0 & 0 & 1\\
																	\end{bmatrix},$};
				\node (m_b) at (2,0) {$M_b = \begin{bmatrix}
																					0 & 0 & 1 & 0\\
																					0 & 0 & 1 & 0\\
																					0 & 0 & 1 & 0\\
																					0 & 0 & 0 & 1\\
																				\end{bmatrix},$};
				\node (m_aa) at (-4,-2) {$M_{aa} = \begin{bmatrix}
																					0 & 0 & 0 & 1\\
																					0 & 0 & 0 & 1\\
																					0 & 0 & 1 & 0\\
																					0 & 0 & 0 & 1\\
																				\end{bmatrix},$};
				\node (m_ab) at (0,-2) {$M_{ab} = \begin{bmatrix}
																						0 & 0 & 1 & 0\\
																						0 & 0 & 0 & 1\\
																						0 & 0 & 1 & 0\\
																						0 & 0 & 0 & 1\\
																					\end{bmatrix},$};																			
				\node (m_ba) at (4,-2) {$M_{ba} = \begin{bmatrix}
																						0 & 0 & 1 & 0\\
																						0 & 0 & 1 & 0\\
																						0 & 0 & 1 & 0\\
																						0 & 0 & 0 & 1\\
																					\end{bmatrix}$};
				\end{tikzpicture}}}
		\captionsetup{skip=5pt}
		\caption{\label{fig:syntactic-monoid} Construction of the syntactic monoid for $L_0$.}
	\end{figure}
\end{example}
Notice that in the above example, the transition matrix for any word of length $\geq 3$ is identical to the transition matrix for some word of length $\leq 2$. Below we formally state this property.
\begin{lemma}\label{lem:almeida}\emph{\cite[\S 3.7.2]{almeida1995finite}}
	Given a PPT language $L = \bigcup_{j\in J} w_j\Sigma^*$, for a finite set of indices $J$, there exists a natural number $k\in\nat$, such that the following hold:
	\begin{enumerate}
		\item the set $\eta_L(\Sigma^k)$ comprises of all the idempotents in $\scrsf{syn}(L)$, and
		\item  all the idempotents in $\scrsf{syn}(L)$ are left-zeroes
	\end{enumerate}
	Moreover, $k = max_{j\in J} |w_j|$ can be computed.
\end{lemma}
Intuitively, this means that the transitive closure of the transition relation, from any state in the corresponding automaton, is dictated by the \emph{first} $k$ positions of the input word.
\begin{example}
	In the transition monoid $\synmon{L_0}$ from Example \ref{ex:syntactic-monoid} (see also Figure \ref{fig:syntactic-monoid}), the idempotents are the elements $M_b = M_{bb},M_{aa},M_{ab}$, and $M_{ba}$. These idempotents also act as left-zeroes. Indeed, from any state, after reading the first two letters of the input word, the automaton reaches one of the sink states $q_2$ and $q_3$.
\end{example}

We will now prove a crucial property specific to a subclass of PPT languages that forms the cornerstone of our domain abstraction in the following sections. 
\begin{lemma}\label{lem:prefix-prefix-singleton-ec}
		Consider the PPT language $L$ of the form $p\cdot\Sigma^+$ or $p\cdot\Sigma^{\geq 2}$, for some $p\in \Sigma^*$. If $u \preceq p$ is a prefix of $p$, then $u \sim_L u' \implies u = u'$. Moreover, if $L$ is of the form $p\cdot\Sigma^{\geq 2}$, and $u \in p\cdot\Sigma$, then $u \sim_L u' \implies u' \in p\cdot\Sigma$.
\end{lemma}
\begin{proof}
	For the first part, assume that $L$ is of the form $p\cdot\Sigma^+$; the proof for the other case is similar. Since $u \preceq p$, for some letter $a$, the string $uu^{-1}pa \in L$. Moreover, since $u \sim_L u'$, $u'(u^{-1}pa) \in L$. We now consider the size of $u'$ and do a case-by-case analysis:
	\paragraph{$|u'| > |p|$}
	It must be that $p \prec u'$, and hence $u' \in L$. But, since $u \sim_L v$ we have $u \in L$, which implies that $p \prec u$. This leads to a contradiction.
	\paragraph{$|u'| \leq |p|$}
	This implies that $u' \preceq p$. Now we have $u \preceq p$, $u' \preceq p$ and $u \sim_L u'$. If $u = u'$, we are done. Otherwise, either $u \prec u'$ or $u' \prec u$. Without loss of generality, assume $u \prec u'$. Then there exists $w$ s.t. $u'w \in p\Sigma \subseteq p\Sigma^+$. Since $u \sim_L u'$, we must have that $uw \in L$. But since $u \prec u'$, $|u| < |u'|$, which implies that $|uw| < |u'w| = |p|+1$. This leads to a contradiction.
	
	For the second part, since $u \in p\cdot\Sigma$, there exists $a\in \Sigma$ such that $u = p \cdot a$. Moreover, there exists $x \in \Sigma^*$ such that $p\cdot a\cdot x \in L$. Since $u \sim_L u'$, $u'\cdot x \in L$. We now consider the size of $u'$ and do a case-by-case analysis:
	\paragraph{$|u'| > |p + 1|$}
	There must exist $b,c\in\Sigma$ and $y\in\Sigma^*$ such that $u' = p\cdot b \cdot c\cdot y$. Using Lemma \ref{lem:almeida}, we know that $\synhom{L}(p\cdot b\cdot c)$ is a left-zero in $\synmon{L}$. Hence, $u' \sim_L p\cdot b\cdot c$. By transitivity, we thus obtain $p\cdot a \sim_L p\cdot b\cdot c$. But clearly, for the empty string $\epsilon$, $p\cdot a\cdot\epsilon \notin L$ but $p \cdot b\cdot c\cdot\epsilon\in L$. This leads to a contradiction.
	\paragraph{$|u'| < |p + 1|$}
	This means that $u' \preceq p$. From the first part, we then know that $u = u'$, but by assumption $|u| = p + 1$. This leads to a contradiction.
	
	Thus, it must be that $|u'| = p+1$, and hence, $u' \in p\cdot\Sigma$.
\end{proof}

The above lemma shows that for a language $L$ of the form $p\cdot\Sigma^+$ (and $p\cdot\Sigma^{\geq 2}$), the set $\synhom{L}^{-1}(\synhom{L}(u))$ is a singleton, for $u \preceq p$, i.e., all prefixes of $p$ are syntactically congruent only to themselves, with respect to $L$. Moreover, for $L = p\cdot\Sigma^+$ (or $p\cdot\Sigma^{\geq 2}$) and $u \preceq p$, $\synhom{L}^{-1}(prec(\synhom{L}(u)))$ is the set of prefixes of $u$.

Before we define our abstraction on the query space, we will first introduce a more generalized notion of bisimulation for systems with LTS semantics. We will then show that our abstraction induces this generalized notion of bisimulation on $\scrsf{DNS}$.

\subsection{A Generalized Bisimulation}\label{subsec:gen-ker-bis}

Recall the definition of an LTS from  \Cref{subsec:trcps}. Here we define a novel generalization of bisimulation on LTSes, called a \emph{generalized kernel bisimulation}. In contrast to standard notions of bisimulation, a generalized kernel bisimulation consists of equivalence relations on both states \emph{and} transitions such that two equivalent states are \emph{bisimilar} up to their equivalent outgoing transitions. In order to formalize this, we adjust Glushkov's  definition of Mealy automata homomorphisms %
\cite[p. 3]{glushkov1961abstract}
to define homomorphisms for LTS.

\begin{definition}[LTS Homomorphism]
	An \emph{LTS homomorphism} of an LTS $\cal A^1 = \tup{S^1,S^1_{init},A^1,\trans^1}$ onto an LTS $\cal A^2 = \tup{S^2,S_{init}^2,A^2,\trans^2}$ is a product map $H = \tup{h_{state},h_{act}}$ of the surjective maps $h_{state} : S^1 \twoheadrightarrow S^2$ and $h_{act} : A^1 \twoheadrightarrow A^2$, such that the following formulae are valid for any $s^1\in S^1$ and $a^1 \in A^1$,
	\begin{align*}
		\forall s'^1.(s^1 \xrightarrow{a^1} s'^1 \in \trans^1 & \implies h_{state}(s^1) \xrightarrow{h_{act}(a^1)} h_{state}(s'^1) \in \trans^2) \\
		s^1 \in S_{init}^1 & \implies h_{state}(s^1) \in S_{init}^2
	\end{align*}
	We write $H : \cal A^1 \rightarrow \cal A^2$ to denote the homomorphism. The set of all homomorphisms is closed under composition. The identity homomorphism is the identity map $id\subsub{\cal A^1} = \tup{id\subsub{S^1},id\subsub{A^1}}$, where $id\subsub{S^1}$ and $id\subsub{A^1}$ are the identity maps on $S^1$ and $A^1$ respectively.
	
	\looseness=-1 Given a run $\rho = s_0 \xrightarrow{a_1} s_1 \ldots s_{n-1} \xrightarrow{a_n} s_n$ in $\cal A^1$, we write $H(\rho)$ to denote its pointwise homomorphic image $h_{state}(s_0) \xrightarrow{h_{act}(a_1)} h_{state}(s_1) \ldots h_{state}(s_{n-1}) \xrightarrow{h_{act}(a_n)} h_{state}(s_n)$, which is a run in $\cal A^2$.
\end{definition}

Essentially, a homomorphism from an LTS $\cal A^1$ to $\cal A^2$ induces an equivalence relation on the set of states and actions. Indeed, we say two states (actions) are equivalent if they have the same image with respect to the function $h_{state}$ ($h_{act}$). Alternatively, $s^1$ ($a^1$) is equivalent to $s'^1$ ($a'^1$) iff $\tup{s^1,s'^1}$ ($\tup{a^1,a'^1}$) is in the \emph{kernel} of the homomorphism $h_{state}$ ($h_{act}$), i.e., $\tup{s^1,s'^1}\in \ker{h_{state}}$ ($\tup{a^1,a'^1}\in \ker{h_{act}}$), where $\ker{f}$ denotes the kernel of a function $f$. Consequently, such an equivalence induces a congruence on the structure $\cal A^1$. We now define a generalized kernel bisimulation on an LTS.

\begin{definition}[Generalized Kernel Bisimulation]\label{def:kernel-bisim}
	Given an LTS $\cal A^1 = \tup{S^1,S^1_{init},A^1,\trans^1}$, a homomorphism $H$ of $\cal A^1$ to $\cal A^2 = \tup{S^2,S_{init}^2,A^2,\trans^2}$ is called a \emph{generalized kernel bisimulation} on $\cal A^1$ if and only if, for every tuple $\tup{p,p'}\in\ker{h_{state}}$ and all actions $a \in A^1$, if $p \xrightarrow{a} q \in \trans^1$, then $p' \xrightarrow{a'} q' \in \trans^1$, such that $\tup{a,a'}\in \ker{h_{act}}$ and $\tup{q,q'}\in \ker{h_{state}}$.
	
	Given two states $p$ and $p'$ in $S^1$, we say $p$ is \emph{kernel bisimilar} to $p'$ with respect to $H$, written $p \sim\subsub{H} p'$, if $H$ is a generalized kernel bisimulation on $\cal A^1$, such that $\tup{p,p'}\in\ker{h_{state}}$.
\end{definition}

If the generalized kernel bisimulation $H$ is such that $h_{act} = id_{A^1}$ is the identity map on $A^1$, then $\ker{h_{state}}$ is a strong bisimulation \cite[Appendix B]{tech-rpt}. Conversely, for any strong bisimulation  $R \subseteq S^1 \times S^1$ on $\cal A^1$, we define the map $H^{R} = \tup{h_{state}^{R},id_{A^1}}$, where $h_{state}^{R} : S^1 \rightarrow S^1/R$ is the quotient map with respect to the equivalence relation $R$. The map $H^R$ is indeed a homomorphism of $\cal A^1$ onto the LTS $\cal A'^2 = \tup{S^1/R,S_{init}^1/R,A^1,\trans'^1}$ where $\trans'^1 = \{ h_{state}^{R}(p) \xrightarrow{a} h_{state}^{R}(q) \mid p \xrightarrow{a} q \in \trans^1 \}$. Moreover, for any $a \in A^1$,
\begin{align*}
	p\sim\subsub{H^R} p' & \Leftrightarrow \tup{p,p'} \in R\\
	& \Rightarrow (p \xrightarrow{a} q \in \trans^1 \Rightarrow \exists q'. p' \xrightarrow{a} q' \in \trans^1 \wedge \tup{q,q'}\in R)\\
	& \Leftrightarrow (p \xrightarrow{a} q \in \trans^1 \Rightarrow \exists q'. p' \xrightarrow{a} q' \in \trans^1\\
	&  \,\,\,\,\wedge \tup{a,a}\in \ker{id\subsub{A^1}} \wedge q\sim\subsub{H^R}q'.
\end{align*}
Hence, generalized kernel bisimulation is weaker than strong bisimulation. Below, we state as a corollary, that any two states that are generalized kernel bisimilar have equivalent sets of traces starting from them, upto their homomorphic image.

\begin{lemma}\label{lem:bisim}
	Consider an LTS $\cal A^1 = \tup{S^1,S^1_{init},A^1,\trans^1}$ and a kernel bisimulation $H : \cal A^1 \rightarrow \cal A^2$. Then for any run $\rho$ in $\cal A^2$ from a state $s'^2$, there exists a run $\rho'$ in $\cal A^1$ from $s^1$, for each $s^1 \in h_{state}^{-1}(s'^2)$, such that $H(\rho') = \rho$.
\end{lemma}
\begin{proof}
	The proof follows by induction, on the length of runs in $\cal A^2$. The base case for runs of length 1 is immediate from the definition of generalized kernel bisimulation.
\end{proof}

Note that kernel bisimulations have already been introduced and well-studied in the context of category theory \cite{goncharov2014coalgebraic,staton2011relating}. Moreover, it is well-known that they coincide with strong bisimulations in the context of LTSes. Thus, generalized kernel bisimulations essentially generalize kernel bisimulations in the context of LTSes. Hereafter, for brevity, we will refer to generalized kernel bisimulations as just kernel bisimulations.

\subsection{Syntactic Domain Congruence}\label{subsec:syn-domain-congruence}

Recall that every query consists of a domain name and a query type.
Since there are finitely many query types, the unbounded nature of
$\scrsf{query}$ arises from the unbounded\footnote{In practice, it is bounded by hardware or pre-defined limits. However, the limits are extremely large. By current standards, the total number of possible queries is $\sum_{i=0}^{253}38^i$, where 253 is the maximum length of a domain name, and each position is labelled by one of 38 valid characters. The total number of possible labels is $\sum_{i=0}^{253}37^i$, as a label is a domain without any periods (.).} cardinality of $\scrsf{domain}$.
This is because every domain name consists of an unbounded
sequence of labels, where each label comes from the
set $\scrsf{label}$ (Example \ref{ex:normal-domain}). In this section, we define a more robust notion of a syntactic congruence over $\scrsf{domain}$ with finitely many equivalence classes that preserves the semantics of $\scrsf{dns}$. In particular, the congruence is preserved not just under concatenation but also domain name rewrites. We shall show that every congruence class is a PPT language over $\scrsf{label}$.

As usual, we fix the zone-file configuration $\cal Z$ for all our computations. Since every domain name is also a string, henceforth, we call a set of domain names a language. To begin with, we build some intuition on how one may abstract the domain names to a finite set.

\subsubsection{Basic Intuition}

As a first approach, we will try to define an equivalence relation on $\scrsf{domain}$ with finitely many equivalence classes, such that any two global configurations in $\scrsf{dns}$ that are equivalent upto their query components, are kernel bisimilar. To do that, we first identify some languages by observing the behaviour of the system,
\begin{align*}
	L_{s, t, \ltimes} &= \{d \mid d\in\scrsf{domain},\text{dn}(\cal M (s,\tup{d,t})) \ltimes d\}\\
	L_{t, rr} &= \{d \mid d\in \scrsf{domain}, d \leftarrow_{t} rr\}
\end{align*}
for $s \in \cal S, t \in \scrsf{type}, \ltimes \in \{=,\ni_{\star},<\},$ and $rr\in \cal Z$. Any two domain names $d$ and $d'$ with identical membership relations with the above languages are defined to be equivalent, i.e., $d$ and $d'$ are equivalent if for each of the above languages, say $L$, $d\in L$ iff $d'\in L$. Intuitively, two global configurations in $\scrsf{dns}$ where the corresponding occurrence of domain names are equivalent (we exclude the unsynthesized resource records from this comparison) according to the above relation will either both satisfy the premise for a transition or neither, for example, observe the premises of transitions in Figures \ref{fig:ns-transitions}, \ref{fig:ns-transitions-more}, and \ref{fig:rec-transitions}. Since the above set of languages is finite, there are finitely many equivalence classes. Readers may check that this also induces a kernel bisimilarity on the global configuration space of $\scrsf{dns}$.

Note that this also suffices to ensure kernel bisimilarity in the case of $\rtype{DNAME}$ rewrites. Indeed, by definition, for two equivalent domain names $d$ and $d'$, and a $\rtype{DNAME}$ record $rr$, either they both contain $\text{dn}(rr)$ as a prefix or neither. This means, either both names can undergo the rewrite operation or neither. After the rewrite, both the resulting domain names will have the same prefix up to all the domain names of the records in $\cal Z$. However, a priori it is not clear whether rewrites also preserve the equivalence of the domain names, as the membership in the languages of the form $L_{s,t,\ltimes}$ may now differ. Moreover, a more fundamental concern is that there is no obvious way to effectively represent the above languages.

\subsubsection{Formal Definition}
We will now define a syntactic congruence on domain names that slightly refines the one sketched above. To begin with, let us denote the set of all labels occuring in the domain names in the records of $\cal Z$ by $\zoneAlph\subseteq\scrsf{label}$. Additionally, we symbolically denote the elements in the complement set $(\scrsf{label}\setminus\zoneAlph)$ by the special label $\zoneAlphCompl$. This finally gives us the finite alphabet, $\dnsAlph = \zoneAlph\cup\{\zoneAlphCompl\}$.

\begin{definition}[Syntactic Domain Congruence]
	For every record $rr \in \cal Z$, consider the PPT languages $L_{rr,\text{dn}}$ and $L_{rr,\text{val}}$, such that
	\noindent\begin{minipage}{.5\linewidth}
		\begin{equation*}
			L_{rr,\text{dn}} = \begin{cases}
				\text{dn}(rr).\dnsAlph^{\geq 2}, & \text{if } \text{dn}(rr) \in (\scrsf{label}^*)_\star\\
				\text{dn}(rr).\dnsAlph^{+},  & \text{if } \text{dn}(rr) \in (\scrsf{label}^*)\\
			\end{cases}
		\end{equation*}
	\end{minipage}%
	\begin{minipage}{.5\linewidth}
		\begin{equation*}
			L_{rr,\text{val}} = \begin{cases}
				\text{val}(rr).\dnsAlph^+,  & \text{if } \text{ty}(rr) = \scrsf{ns}\\
				\dnsAlph^*,  & \text{if } \text{ty}(rr) \neq \scrsf{ns}.
			\end{cases}
		\end{equation*}
	\end{minipage}

	The \emph{syntactic domain congruence} with respect to $\cal Z$ is the congruence relation $\simZ$ on $\dnsAlph^*$ s.t.
	\[ u \simZ v \iff \forall rr \in \cal Z. (\synhom{L_{rr,\text{dn}}}(u) = \synhom{L_{rr, \text{dn}}}(v) \wedge \synhom{L_{rr,\text{val}}}(u) = \synhom{L_{rr, \text{val}}}(v)). \]
	We then say that the domains $u$ and $v$ are \emph{syntactically congruent} (with respect to $\cal Z$).
	
	Equivalently, $u \simZ v$ iff $\synhomZ (u) = \synhomZ (v)$, where $\synhomZ (w) = \tup{\synhom{L_{rr,j}}(w)}_{rr\in\cal Z,j\in \{\text{dn},\text{val}\}}$. The co-domain of $\synhomZ$ is the product monoid $\prod_{rr\in\cal Z,j\in \{\text{dn},\text{val}\}} \synmon{L_{rr,j}}$, denoted by $\synmonZ$. %
\end{definition}

Intuitively, the idea is to define a congruence that is sufficiently refined such that two domain names satisfy the conditions required to trigger a transition in $\scrsf{dns}$ if and only if they are congruent to each other; the languages must be carefully chosen to make this happen. Below we show that this is indeed the case. In particular, we show that the syntactic domain congruence refines the monadic relations, defined as languages, from our dicussion above.

\begin{lemma}\label{lem:synt-dom-cong-fundamental}
	Consider two domain names $d$ and $d'$ such that $d \simZ d'$, then the following hold:
	\begin{enumerate}
		\item $\forall rr\in\cal Z. (\text{dn}(rr) \ltimes d \iff \text{dn}(rr) \ltimes d')$, where $\ltimes \in \{=,\ni_{\star},<\}$.
		\item $\forall t\in\scrsf{type}. \scrsf{rank}(rr,\tup{d,t},z) = \scrsf{rank}(rr,\tup{d',t},z)$, where $z$ is an arbitrary zone in $\cal Z$ and $rr$ is an arbitrary record in $z$.
		\item $\forall rr\in\cal Z. \text{ty}(rr) = \scrsf{ns} \implies (\text{val}(rr) = d \iff \text{val}(rr) = d')$.
	\end{enumerate}
\end{lemma}
\begin{proof}
	
	(1) Without loss of generality, we shall do a case-by-case analysis of the domain name $d$, generalizing $rr \in \cal Z$,
	\paragraph{$(\ltimes \equiv\,=)$}
	Trivially, $d = \text{dn}(rr)$ implies $d \preceq \text{dn}(rr)$. Since $d \simZ d'$, then $d \sim_{L_{rr,\text{dn}}} d'$. Using Lemma \ref{lem:prefix-prefix-singleton-ec}, we thus have $d = d'$, which implies $d' = \text{dn}(rr)$.
	\paragraph{$(\ltimes \equiv\ni_\star)$}
	Recalling the definition of $\ni_\star$ from Figure \ref{fig:DNS-notation}, we know that $d_\star \in \text{dn}(rr)\ostar{(\dnsAlph)}_\star$ and $\text{dn}(rr) \in (\scrsf{label}^*)_\star$. By definition, $L_{rr,\text{dn}} = \text{dn}(rr).\dnsAlph^{\geq 2}$. Since $d \sim_{L_{rr,\text{dn}}} d'$, using Lemma \ref{lem:prefix-prefix-singleton-ec}, we again have that $d' \in \text{dn}(rr)\cdot{(\dnsAlph)}$, or equivalently, $d'_\star \in \text{dn}(rr)\ostar{(\dnsAlph)}_\star$ (recall the definition of $\cal L_\star$ from~\Cref{subsec:def-not}).
	\paragraph{$(\ltimes \equiv\,<)$}
	Recalling the definition of $<$ from Figure \ref{fig:DNS-notation}, $\text{dn}(rr) \subset prec (d)$. Since the underlying monoid is a free monoid, this is equivalent to $\text{dn}(rr) \prec d$. Using Lemma \ref{lem:almeida}, $\synhom{L_{rr,\text{dn}}}(d)$ is thus a left-zero in $\synhom{L_{rr,\text{dn}}}$, which implies that $\synhom{L_{rr,\text{dn}}}(d')$ is also a left-zero in that semigroup. Hence, using Lemma \ref{lem:prefix-prefix-singleton-ec}, $\text{dn}(rr) \prec d'$, or equivalently, $\text{dn}(rr) \subset prec (d')$.
	
	(2) It suffices to prove that $\scrsf{match}(rr,\text{dn}(rr)) = \scrsf{match}(rr,d) \iff \scrsf{match}(rr,\text{dn}(rr)) = \scrsf{match}(rr,d')$ and $max_\simeq(\text{dn}(rr), d) = max_\simeq(\text{dn}(rr), d')$. The first formula follows from the arguments in (1).\\
	To prove the second formula, we must prove that $|prec(\text{dn}(rr)) \cap prec(d)| = |prec(\text{dn}(rr)) \cap prec(d')|$. Since we are working with the free monoid $\cal L$, $prec(\text{dn}(rr)) \cap prec(d)$ essentially contains the prefix-closed set of the largest common prefix of $\text{dn}(rr)$ and $d$. We will show that any common prefix of $\text{dn}(rr)$ and $d$ must be a prefix of $d'$.\\
	Suppose $d_1$ is a common prefix of $\text{dn}(rr)$ and $d$. Then there must exist an $m\in \synmon{L_{rr,\text{dn}}}$ such that, $\synhom{L_{rr,\text{dn}}}(d_1)\cdot m = \synhom{L_{rr,\text{dn}}}(d) = \synhom{L_{rr,\text{dn}}}(d')$.Hence, there is a prefix $d_2$ ,of $d'$, in $\synhom{L_{rr,\text{dn}}}^{-1}(d_1)$. However, since $d_1$ is a prefix of $\text{dn}(rr)$, it follows from Lemma \ref{lem:prefix-prefix-singleton-ec}, that $d_1 = d_2$. Thus, $d_1$ is a prefix of $d'$ as well.
	
	(3) This follows immediately from the definitions and Lemma \ref{lem:prefix-prefix-singleton-ec} using similar arguments as before.
\end{proof}

The following theorem is a direct corollary of the above lemma,
\begin{theorem}
	Consider two domain names $d$ and $d'$ such that $d \simZ d'$, then for every $s\in\cal S$, $t\in\scrsf{type}$, $\ltimes \in \{=,\ni_{\star},<\}$ and $rr\in \cal Z$, it holds that
	\begin{enumerate}
		\item $d \in L_{s,t,\ltimes} \iff d' \in L_{s,t,\ltimes}$
		\item $d \in L_{t,rr} \iff d' \in L_{t,rr}$.
	\end{enumerate}
\end{theorem}

Below, we state that the syntactic domain congruence also commutes with $\rtype{DNAME}$ rewrites.
\begin{lemma}
	Consider a $\rtype{DNAME}$ resource record $\tup{pre,\scrsf{dname},\sigma, post,b}$, and two syntactically congruent domain names $d$ and $d'$. Then $(post.(pre)^{-1}.d) \simZ (post . (pre)^{-1} . d')$.
\end{lemma}
\begin{proof}
	It suffices to show that $pre \preceq d$ iff $pre \preceq d'$. Since, $pre = \text{dn}(rr)$, for some resource record $rr$, this follows immediately from Lemma \ref{lem:synt-dom-cong-fundamental} (1).
\end{proof}

The congruence can be naturally extended to a (non-symbolic) domain name $d \in \scrsf{domain}$, where any label in $(\scrsf{label}\setminus\zoneAlph)$ is mapped to $\synhomZ(\zoneAlphCompl)$. Additionally, the root domain $\epsilon$ is mapped to $\synhomZ (\varepsilon)$. Thus, we obtain a monoid homomorphism of the free monoid $\cal L$. We can effectively construct the equivalence classes of the entire set of domain names occuring in $\cal Z$, as shown in the following lemma.

\begin{lemma}
	Assume that $\cal Z$ has $n$ resource records and that the maximum length of the domain name of a resource record in $\cal Z$ is $c$. Then the class of $ $ abstract domains can be computed in $\cal O (n^{c+1}c^5)$ time.
\end{lemma}
\begin{proof}
    Each resource record in the zone file configuration corresponds to at most two PPT languages, $L_{rr,\text{dn}}$ and $L_{rr,\text{val}}$. There at most $\cal O (n)$ such languages. Due to the nature of these languages, one can show that to compute the corresponding automata we require $\cal O (|\text{dn}(rr)|)$ time. The adjacency matrix for each letter in the alphabet $\dnsAlph$ can be initiated in $\cal O(|\dnsAlph||\text{dn}(rr)|^2)$ time. To compute adjacency matrices, we thus need $O(|\dnsAlph||\text{dn}(rr)|^2)$ time.

    To compute the set of abstract domains, we take the diagonal product over all the languages, for each label, and then compute all concatenations up to the length $c$. Since there are at most $\cal O (n)$ matrices for each label (since there are at most that many languages), we will compute $|\dnsAlph|^cnc$ matrix multiplications in all. Each multiplication takes $\cal O(c^{\sim 3})$ time because the dimensions of a matrix are at most $c\times c$. Therefore, to compute all the domains we need $\cal O(|\dnsAlph|^cnc^4)$ time.

    However, since $|\dnsAlph| \leq nc$, we thus have $\cal O(|\dnsAlph|^cnc^4) = \cal O(n^{c+1}c^5)$.
\end{proof}

\subsection{Abstract DNS Model}\label{subsec:abstract-model}

Recalling Example \ref{ex:normal-domain}, we have defined an \emph{abstraction} on $\scrsf{domain}_{\cal L}$ with respect to $\cal Z_{\cal L}$, induced by a congruence on $\dnsAlph^*$ such that the quotient structure forms a monoid. The abstraction can be computed using the monoid homomorphism $\synhomZ$.

Using Lemma \ref{lem:notation-functorial}, the homomorphism $\synhomZ$ can be functorially mapped to the map $Z : \scrsf{zfc}\subsub{\cal L} \rightarrow \scrsf{zfc}\subsub{\synmonZ}$. We also observe that, for a zone file configuration $\cal Z' \in \scrsf{zfc}\subsub{\monoid}$ over the monoid $\monoid$, the semantics of the DNS instance $\tup{\scrsf{dns}\subsub{\cal Z'}}$ is functorial in $\monoid$. This is a straightforward extension of Lemma \ref{lem:notation-functorial}, following the constructions described in Sections \ref{subsec:NS-semantics} and \ref{subsec:RR-semantics}. Now, the arbitrarily fixed instance $\scrsf{dns}$, used in Sections \ref{subsec:NS-semantics} and \ref{subsec:RR-semantics}, is defined over the monoid $\cal L$. Hence, the homomorphism $\synhomZ : \cal L \rightarrow \synmonZ$ can be functorially mapped to an LTS homomorphism, $H_{\synhomZ} : \tup{\scrsf{dns}\subsub{\cal Z}} \rightarrow \tup{\abstraction}$. We call $\abstraction$ the \emph{abstraction} of the concrete model $\scrsf{dns}\subsub{\cal Z}$.

\Cref{fig:dns-abstraction}  shows the \emph{abstract} nameserver transitions in $\abstraction$ from Figures \ref{fig:ns-transitions} and \ref{fig:ns-transitions-more}. Intuitively, in all the states and transitions, every occurrence of a domain name is replaced by its abstraction. %

\begin{figure}
	\centering
	\subfloat{
		\centering
		\resizebox{0.35\textwidth}{!}{
			\begin{tikzpicture}[every label/.style={font=\small}]
				\node (before) [fill=blue!10,draw,rounded corners=2pt, label=north:\emph{Concrete Config}] at (0,0) {$\dfrac{\text{dn}(\cal M (s, \tup{d,t})) = d}{\tup{lookup,\tup{d,t},r} \rightarrow \tup{exact, \tup{d,t}, r}}$};
				\node (after) [fill=olive!10,draw,rounded corners=2pt, label=south:\emph{Abstract Config}] at (0,-2) {$\dfrac{\text{dn}(\cal M (s, \tup{\synmonZ(d),t})) = \synmonZ(d)}{\tup{lookup,\tup{\synhomZ (d),t},r} \rightarrow \tup{exact, \tup{\synhomZ (d),t}, r}}$};
				\draw [double distance=2pt, -{Implies}, shorten <= 2pt, shorten >= 2pt] (before) -- (after) node [midway, right] {$\synhomZ$};
			\end{tikzpicture}}}\hfill
	\subfloat{
		\centering
		\resizebox{0.375\textwidth}{!}{
			\begin{tikzpicture}[every label/.style={font=\small}]
				\node (before) [fill=blue!10,draw,rounded corners=2pt, label=north:\emph{Concrete Config}] at (0,0) {$\dfrac{\text{dn}(\cal M (s, \tup{d,t})) \ni_{\star} d}{\tup{lookup,\tup{d,t},r} \rightarrow \tup{wildcard, \tup{d,t}, r}}$};
				\node (after) [fill=olive!10,draw,rounded corners=2pt, label=south:\emph{Abstract Config}] at (0,-2) {$\dfrac{\text{dn}(\cal M (s, \tup{\synhomZ (d),t})) \ni_{\star} \synhomZ (d)}{\tup{lookup,\tup{\synhomZ (d),t},r} \rightarrow \tup{wildcard, \tup{\synhomZ (d),t}, r}}$};
				\draw [double distance=2pt, -{Implies}, shorten <= 2pt, shorten >= 2pt] (before) -- (after) node [midway, right] {$\synhomZ$};
	\end{tikzpicture}}}\qquad
	\subfloat{
		\centering
		\resizebox{\textwidth}{!}{
			\begin{tikzpicture}[every label/.style={font=\small}]
				\node (before) [fill=blue!10,draw,rounded corners=2pt, label=north:\emph{Concrete Config}] at (0,0) {$\dfrac{
						\text{dn}({\cal M}(s,\tup{d,t}))<d,\scrsf{dname}\in\tau(s,\tup{d,t}), d'\in\scrsf{domain}, {\cal T}_{\scrsf{dname}}(s,\tup{d,t})=\{rr\}, rr'=\tup{d,\scrsf{cname},\text{ttl}(rr),\text{val}(rr)\text{dn}(rr)^{-1}d,1}}{\tup{lookup,\tup{d,t},r}\xrightarrow{c_{(d',s),r}!\tup{\scrsf{ansq},\{rr,rr'\}}}\tup{init,\bot,\bot}}$};
				\node (after) [fill=olive!10,draw,rounded corners=2pt, label=south:\emph{Abstract Config}] at (0,-2.375) {$\dfrac{
						\text{dn}({\cal M}(s,\tup{\synhomZ (d),t})) \in prec(\synhomZ (d)),\scrsf{dname}\in\tau(s,\tup{\synhomZ (d),t}), m'\in\synmonZ, {\cal T}_{\scrsf{dname}}(s,\tup{\synhomZ (d),t})=\{\synhomZ (rr)\}, rr'=\tup{\synhomZ (d),\scrsf{cname},\text{ttl}(\synhomZ (rr)),m',1}}{\tup{lookup,\tup{\synhomZ(d),t},r}\xrightarrow{c_{(m',s),r}!\tup{\scrsf{ansq},\{\synhomZ (rr),rr'\}}}\tup{init,\bot,\bot}}$};
				\draw [double distance=2pt, -{Implies}, shorten <= 2pt, shorten >= 2pt] (before) -- (after) node [midway, right] {$\synhomZ$};
	\end{tikzpicture}}}
	\captionsetup{skip=8pt}
	\caption{\label{fig:dns-abstraction}Some abstract $\ns{s}$ transitions; the concrete configurations are written before the arrows and the abstract ones after. For the abstract configuration in the bottom-most transition, we assume $m' \in \text{val}(\synhomZ(rr))\text{dn}(\synhomZ(rr))^{-1}\synhomZ(d)$.}
\end{figure}

We conclude this section with one of the main contributions of this paper, establishing both the soundness and completeness of our abstraction.

\begin{theorem}[Soundness and Completeness]\label{thm:sound-complete-abstraction}
	For every run $\rho$ in $\tup{\abstraction}$, there exists a run $\rho'$ in $\tup{\scrsf{dns}\subsub{\cal Z}}$ such that $H_{\synhomZ}(\rho') = \rho$. The converse holds as well.
\end{theorem}
\begin{proof}
	It suffices to show that the LTS homomorphism $H_{\synhomZ}$  is a kernel bisimulation. The statement then follows from Lemma \ref{lem:bisim}.
\end{proof} 

This theorem shows that the set of runs in the abstraction $\abstraction$ is neither an over-approximation nor an under-approximation, but  captures the exact set of runs in $\scrsf{dns}\subsub{\cal Z}$, up to the homomorphic image with respect to $H_{\synhomZ}$. Using Lemma \ref{lem:eager-dns}, we further have the following corollary.

\begin{corollary}\label{cor:eager-abstraction}
	$\abstraction$ is eager.
\end{corollary}

With this, we have obtained a sound and complete abstraction of $\scrsf{dns}$ in the form of a trCPS with finitely many global states. In the rest of this paper, runs refer to the runs in the abstraction $\abstraction$, unless otherwise stated. As the next step, we reduce the semantics of $\abstraction$ to its \emph{untimed model}.

\subsubsection*{Abstracting Time}
The real-time semantics of the model makes the global configuration space infinite. Since the timer semantics in our model only involve resetting the timer ($[\mathbf{x} \leftarrow \mathbf{k}]$), by checking the expiration of the timer ($x > 0$) and devolution of time ($\xrightarrow{t}$), this immediately allows us to employ a standard region construction \cite{alur1994theory} to our model in order to obtain an untimed symbolic reduction of $\abstraction$. 

As the reachability relation of configurations is preserved under region construction, Theorem \ref{thm:sound-complete-abstraction} also holds for the reduced untimed model. By abuse of notation, we refer to this model as $\abstraction$.

	\section{\label{sec:property-validation-to-reachability}The DNS Verification Problem}
In this section, we  study the decidability of the DNS verification problem for $\scrsf{dns}\subsub{\cal Z}$. Borrowing notation from  \Cref{subsec:dns-semantics}, we formally define our problem $\verif$ as follows,

\paragraph{\textbf{Input}}
A ZFC $\tup{\cal Z,\Theta} \in \scrsf{zfc}\subsub{\cal L}$ over the free monoid $\cal L$, and a $\mathsf{FO}(<)$ property $\phi$ over strings in $\scrsf{a}^*$, given as the corresponding minimal automaton $\cal F^\phi$, where $\scrsf{a} = \scrsf{raction}_r \uplus \biguplus_{s \in \cal S} \scrsf{saction}_s$.
\paragraph{\textbf{Output}}
The decision of the verification problem for $\tup{\scrsf{dns}\subsub{\cal Z}}$, and $\cal F^{\synhomZ(\phi)}$ over $\synhomZ(\scrsf{a})^*$, \new{along with a witness trace}, where $\cal F^{\synhomZ(\phi)}$ is an FA obtained from $\cal F^\phi$ by mapping every transition label $a \in \scrsf{a}$ to  $\synhomZ(a)$. We denote the corresponding problem instance by $\verif(\cal Z, \Theta, \phi)$. 

\new{Note that the $\verif$ problem does not target the correctness of the DNS protocol itself, but rather an instantiation of the protocol for a given zone configuration. 
	The problem aims to reflect the situation where network administrators already maintain a zone file configuration and would like to  find vulnerabilities in the resulting DNS behaviour.} Intuitively, we want to capture security vulnerabilities as traces of runs in $\scrsf{dns}\subsub{\cal Z}$ characterized by $\mathsf{FO}(<)$-definable
languages. We will look at some common attack vectors in Section \ref{sec:applications}.

\subsection{The Verification Approach}

Recall from Section \ref{subsec:dns-semantics} that the untimed model $\abstraction$ is an rCPS over the pointed topology $\mathsf T\subsub{\cal Z}$. Using Lemma \ref{lem:rcps-to-pds} and Corollary \ref{cor:eager-abstraction}, we can construct, in polynomial time, a pushdown system $\cal P\subsub{\cal Z}$, such that for any run in $\abstraction$ from an $\mathbf{s}$-configuration to an $\mathbf{s}'$-configuration, there is a trace equivalent run in $\cal P\subsub{\cal Z}$, also from an $\mathbf{s}$-configuration to an $\mathbf{s}'$-configuration. It then immediately follows from Theorem \ref{thm:sound-complete-abstraction} that:

\begin{lemma}
	$\scrsf{verif}(\cal Z, \Theta,  \phi)$ reduces to the verification problem for the PDS $\cal P\subsub{\cal Z}$ and the FA $\cal F^{\synhomZ(\phi)}$.
\end{lemma}

\new{It is well-known \cite[Alg. 3]{esparza2000efficient} that the verification problem for a PDS and an FA is in $\mathsf{PTime}$ and $\mathsf{PSpace}$ respectively. Moreover, it is possible to provide a witness trace that satisfies the said property in doubly-exponential time \cite[Section 3.1.4]{reps2005weighted}.} Since $\cal P\subsub{\cal Z}$ is exponential in the size of the input ZFC (due to the exponential blow-up of the local state space of the resolver $\rec{r}$),\footnote{Note that this does not change even after the abstraction as the blow-up appears when we enumerate the caches and server lists in the resolver states.} we have the following result:

\begin{lemma}
    The set of reachable configurations in $\cal P\subsub{\cal Z}$ can be computed in exponential time in the size of the input $\tup{\cal Z, \Theta}$.
\end{lemma}

We define the \emph{product} of $\cal F^{\synhomZ(\phi)}$ with $\cal P_{\cal Z}$ as the standard product PDS $\cal F^{\synhomZ(\phi)} \times \cal P\subsub{\cal Z}$, such that $(p_1,p_2) \xrightarrow{\tup{\alpha, a, \alpha'}} (q_1,q_2)$ is a transition in $\cal F^{\synhomZ(\phi)} \times \cal P\subsub{\cal Z}$ if and only if $p_1 \xrightarrow{a} q_1$ is a transition in $\cal F^{\synhomZ(\phi)}$ and $p_2 \xrightarrow{\tup{\alpha, a, \alpha'}} q_2$ is a transition in $\cal P\subsub{\cal Z}$. The set of initial states is the product of the set of initial states in $\cal F^{\synhomZ(\phi)}$ and $\cal P\subsub{\cal Z}$. We then call a state $\tup{p_1,p_2}$ in $\cal F^{\synhomZ(\phi)} \times \cal P\subsub{\cal Z}$ \emph{bad} if $p_1$ is a final state in $\cal F^{\synhomZ(\phi)}$. It is straightforward to see that the verification problem for $\cal P\subsub{\cal Z}$ and $\cal F^{\synhomZ(\phi)}$ is equivalent to checking the membership of an $s$-configuration, for a bad state $s$, in the set of reachable configurations of $\cal F^{\synhomZ(\phi)} \times \cal P\subsub{\cal Z}$. Thus, we have the complexity result,

\begin{lemma}\label{lem:complexity}
    \new{$\verif(\cal Z, \Theta, \phi)$ is $\mathsf{2ExpTime}$. The complexity reduces to $\mathsf{ExpTime}$ if the output of the problem does not require a witness trace.}
\end{lemma}

\begin{figure}
    \centering
    \begin{tikzpicture}[every label/.style={font=\small}]
        \node (inp) [] at (-1.25, -0.5) {\footnotesize $Input$:};
        \node (zfc) [fill=red!10,draw,rounded corners=2pt] at (0,0) {$\tup{\cal Z, \Theta}$};
        \node (zfcabstract) [fill=red!10,draw,rounded corners=2pt] at (2.5,0) {$\tup{Z(\cal Z), \Theta}$};
        \node (dns) [fill=red!10,draw,rounded corners=2pt] at (5.25,0) {$\abstraction$};
        \node (pds) [fill=red!10,draw,rounded corners=2pt] at (7.25,0) {$\cal P\subsub{\cal Z}$};
        \node (propabs) [fill=purple!10,draw,rounded corners=2pt] at (-0.25,-1) {$\cal F^{\phi}$};
        \node (fa) [fill=purple!10,draw,rounded corners=2pt] at (6.95,-1) {$\cal F^{\synhomZ(\phi)}$};
        \node (prod) [draw,rounded corners=2pt] at (9,-0.5) {\footnotesize $\cal P\subsub{\cal Z}\times\cal F^{\synhomZ(\phi)}$};
        \node (final) [draw, rounded corners=2pt, align=center, purple] at (11.25,-0.5) {\footnotesize Is a bad state\\ \footnotesize reachable?};
        \draw [->] (zfc) -- (zfcabstract) node [midway, above, align=center] {\scriptsize Abstract\\\scriptsize ZFC};
        \draw [->] (zfcabstract) -- (dns) node [midway, above, align=center] {\scriptsize Construct\\ \scriptsize $\scrsf{dns}$};
        \draw [->] (dns) -- (pds) node [midway, above, align=center] {\scriptsize Reduce to\\ \scriptsize PDS};
        \draw [->] (pds) -- (prod);
        \draw [->] (propabs) -- (fa) node [midway, below, align=center] {\scriptsize Abstract FA};
        \draw [->] (fa) -- (prod);
        \draw [double distance=2pt, -{Implies}] (prod) -- (final);
    \end{tikzpicture}
    \caption{\new{Overview of our decision procedure.}}
    \label{fig:procedure}
\end{figure}

\new{Figure~\ref{fig:procedure} provides an overview of our decision procedure. Since our abstraction is effectively defined as a monoid homomorphism ($\synhomZ$), one can deduce a concrete trace by mapping back every configuration in the abstract output (trace) to a canonical concrete configuration. In particular, it suffices to replace every occurrence of a congruence class, under $\synhomZ$, in a configuration by its canonical element. Since every congruence class can be effectively represented by a regular language, the canonical element can be chosen as one of the shortest domain names in the language.\footnote{This choice can be made deterministic by defining an appropriate lexicographic ordering.} Most importantly, the soundness of this procedure follows from the completeness of our abstraction.}

	\section{Applications}\label{sec:applications}

In  this section, we show some applications of our results in the context of a DoS analysis of DNS. We present the examples informally, showing how instances of security properties can be characterized as $\mathsf{FO}(<)$-properties, and translating the security analysis to an instance of the $\verif$ problem. We further demonstrate how we can compute the abstractions of the inputs of the $\verif$ instance using syntactic domain congruence.

\begin{figure}
    \subfloat[\new{An illustration of the NXNSAttack.}]{
    \label{subfig:nxns-illustration}
    \centering
   \resizebox{0.6\textwidth}{!}{\includegraphics[width=\textwidth]{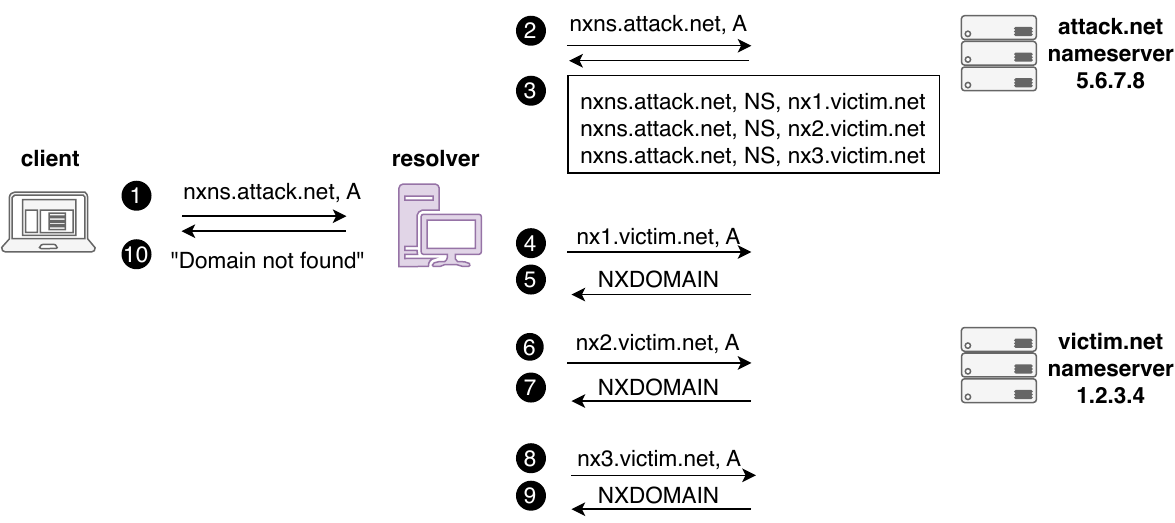}}
    }
\hfill
    \subfloat[\new{The zone configuration. We omit the redundant components (for this attack).}]{
    \label{subfig:nxns-zone}
    \centering
   \resizebox{0.35\textwidth}{!}{
    \begin{tikzpicture}[every label/.style={font=\small}]
				\node (before) [fill=blue!5,draw,rounded corners=2pt, label=north:\emph{Concrete Zone}] at (0,0) {
					\begin{tabular}{l} 
						--- \textbf{zone} \emph{net.attacker} ---\\
						$\tup{net.attack.nxns,\scrsf{ns},net.victim.nx1}$\\
						$\tup{net.attack.nxns,\scrsf{ns},net.victim.nx2}$\\
						$\tup{net.attack.nxns,\scrsf{ns},net.victim.nx3}$\\
					\end{tabular} };
				\node (after) [fill=olive!10,draw,rounded corners=2pt, label=south:\emph{Abstract Zone}] at (0,-3) {					
				\begin{tabular}{l} 
						--- \textbf{zone} $\synhom{\scrsf{nxns}}(net.victim)$ ---\\
						$\tup{\synhom{\scrsf{nxns}}(net.attack.nxns),\scrsf{ns},\synhom{\scrsf{nxns}}(net.victim.nx1)}$\\
						$\tup{\synhom{\scrsf{nxns}}(net.attack.nxns),\scrsf{ns},\synhom{\scrsf{nxns}}(net.victim.nx2)}$\\
						$\tup{\synhom{\scrsf{nxns}}(net.attack.nxns),\scrsf{ns},\synhom{\scrsf{nxns}}(net.victim.nx3)}$\\
				\end{tabular} };
				\draw [double distance=2pt, -{Implies}, shorten <= 2pt, shorten >= 2pt] (before) -- (after) node [midway, right] {$\synhom{\scrsf{nxns}}$};
	\end{tikzpicture}}
    }
\hfill
    \subfloat[
    	A local run in the process $\rec{r}$ in $\scrsf{dns}\subsub{\scrsf{nxns}}$ that corresponds to Figure \ref{subfig:nxns-illustration}. In the initial configuration, the server list $sl_0 = (net.attacker)\odot (net.victim)$ contains the domain names of the other two nameservers, and the initial cache $cache_0 = \{\tup{net.attacker,\scrsf{a},5.6.7.8},\tup{net.victim,\scrsf{a},1.2.3.4}\}$ contains the IP addresses of both these nameservers. For notational convenience, we use $('')$ in a configuration to represent the corresponding component unchanged from the previous configuration; we use $(*)$ to represent a timer valuation after an arbitrary elapse of time; the transitions with labels in \{\ding{183},...,\ding{190}\} informally correspond to the messages in \Cref{subfig:nxns-illustration}, omitting the associated stack operation annotations.  We also drop the messages \ding{182} and \ding{191} as we do not model the client in our model. Instead, the message \ding{182} is represented by the internal transition $\tup{\epsilon, in, \tup{net.example,\scrsf{a}}}$. The messages \ding{186},\ding{188},\ding{190} are replaced by $\tup{\scrsf{nxdomain},\emptyset}$ in the model.
    	]{
    \label{subfig:nxns-run}
    \begin{minipage}[b]{\linewidth}
    \centering
       \includegraphics[width=\textwidth, page=2]{img/formal-run-nxns-crop.pdf}
    \end{minipage}
    }
    \captionsetup{skip=5pt}
\caption{\label{fig:nxns}\new{NXNSAttack.}}
\end{figure}

\subsection{Amplification Attacks}\label{subsec:amplification}

An \emph{amplification attack} is a type of DoS attack where an attacker exploits DNS to overwhelm a target, such as  resolvers and nameservers, with a large volume of traffic, effectively disrupting its ability to function.
For example,
in the recently discovered NXNSAttack (NoneXistent Name Server Attack)~\cite{NXNSAttack}, a small DNS query triggers the DNS resolver to issue a much larger response, potentially achieving a packet amplification factor (PAF) of over a thousand times  the original query size.

Let us first have a closer look at amplification attacks. Informally, given a ZFC $\tup{\cal Z, \Theta}$, a target process $p$, and an amplification factor $\mathit{paf} \in \nat$, the corresponding instance of DNS exhibits an amplification attack if there exists some query that, during the query resolution process, causes the target process $p$ (either a nameserver or a resolver) in the network to receive more than $\mathit{paf}$ messages. Typically, the target process refers to a nameserver in the network which, for a sufficiently large $\mathit{paf}$, is then unable to respond to legitimate query requests. In order to discover amplification attacks, one must exhaustively analyze all possible DNS behaviours, \emph{for every query}, and check if there exists a behaviour where a target nameserver receives more than $\mathit{paf}$ messages.

Borrowing notation from earlier sections, we can formally characterize the set of traces that exhibit an amplification attack as the formula,\footnote{We use many-sorted quantification in the formula, in order to quantify over the space of all messages in $\mathsf{M}\subsub{\cal Z}$, which is an infinite set. The alternative would be an infinite disjunction over all possibilities of $\mathit{msg}$.}

\begin{equation*}\label{eqn:amp-prop}
\resizebox{\textwidth}{!}{$\phi_{amp}(\mathit{paf},p) = \exists x_1,\ldots,x_{\mathit{paf}}\in \nat.\, \left(\bigwedge\limits_{i\neq j \in \{1,\ldots ,\mathit{paf} \}} x_i \neq x_j\right)  \wedge \left( \bigwedge\limits_{i\in \{1,\ldots , \mathit{paf} \}}\bigvee\limits_{j \in P\subsub{\cal Z}} \left(\exists \mathit{msg}\in \mathsf M\subsub{\cal Z}.(\ch_{(j,p)}?\mathit{msg})(x_i)\right) \right)$}
\end{equation*}

The above formula captures the property that there exist $\mathit{paf}$ distinct positions in a trace, such that every one of these positions corresponds to a receive action for the target process $p$. Indeed, any run in $\scrsf{dns}\subsub{\cal Z}$ whose trace satisfies the above property exhibits a receive action for the target component $\cal P^p$ at least $\mathit{paf}$ times. Since every run in $\scrsf{dns}_{\cal Z}$ begins with a query request, a run that satisfies $\phi_{amp}(\mathit{paf},p)$ corresponds to an amplification attack.

\new{
\begin{example}
 Figure~\ref{subfig:nxns-illustration} illustrates how the NXNSAttack is mounted using a simple example.
 This example involves two attacker components, namely a client and an authoritative nameserver \name{attacker.net}, and two innocent components, namely the recursive resolver and an authoritative nameserver \name{victim.net}. 
The client issues a request 
for the subdomain \name{nxns.attack.net} that
is authorized by the attacker's authoritative nameserver (\ding{182}). 
The resolver then attempts to  resolve the subdomain  (\ding{183}).
In response, the attacker's nameserver 
\name{attacker.net}
sends an \rtype{ns} 
referral response with multiple nameserver names (three in this case) but without their IP addresses 
(\ding{184}).
This forces the resolver to issue additional  subqueries to resolve each of these names (\ding{185}--\ding{190}), thus amplifying the DNS traffic.

The ZFC can be represented by $\tup{\scrsf{nxns},\{5.6.7.8,1.2.3.4\}}$, where the domain of $\scrsf{nxns}$ is the set $\{5.6.7.8,1.2.3.4\}$, such that $5.6.7.8$ is the IP of the attacker's nameserver and $1.2.3.4$ that of the victim. The function $\scrsf{nxns}$ essentially\footnote{For simplicity, we only consider the records that are relevant for the attack.} maps $1.2.3.4$ to the empty set and $5.6.7.8$ to the concrete zone illustrated in \Cref{subfig:nxns-zone}. The malicious zone contains three $\scrsf{ns}$ records that point to the nameservers with domain names $net.victim.nx\{i\}$, for $i \in \{1, 2, 3\}$. But, observe that none of them have an accompanying record of type $\scrsf{a}$ that points those names to an IP address. The IP addresses $1.2.3.4$ and $5.6.7.8$ are recorded as the root nameserver IP addresses. 

We now describe the run of the resolver in $\scrsf{dns_{nxns}}$. The resolver starts from the initial configuration, and takes the internal transition with the label $\tup{\epsilon, in, \tup{net.attacker.nxns,\scrsf{a}}}$ to model that it received the query $\tup{net.attacker.nxns,\scrsf{a}}$ from the client. When it \emph{receives} the query, it moves to the $query$ state ready to start resolution. The current stack configuration is $\tup{net.attack.nxns,\scrsf{a}}$, which denotes that the initial query is currently under the resolution. It then sends the query out to the leftmost nameserver in the server list that contains the domain name of the query in its zone, namely $attacker.net$, and then waits for an answer by moving to the state $wait$. Once it receives the response with a set of  $\scrsf{ns}$ records from $\ns{5.6.7.8}$, it adds them to its cache, along with the corresponding domain name of the nameserver in the value of each record (in the answer) to its server list. It moves back to the $query$ state to decide the next step for the resolution. Since every $\scrsf{ns}$ record in the answer was glueless, the resolver proceeds with the resolution of $\tup{net.victim.nx1,\scrsf{a}}$ in order to find the glue. Consequently, it pushes the query on its stack while staying in the $query$ state. Since the query name is in the zone $victim.net$, it sends the query to $\ns{1.2.3.4}$. However, since the zone file for $victim.net$ is empty, the resolver gets $\tup{\scrsf{nxdomain},\emptyset}$ (\ding{186}) as the answer, thus deleting $net.victim.nx1$ from its server list and moving back to the $query$ state. This repeats two more times for the other two entries in the server list. It is then straightforward to see that this behaviour leads to $\ns{1.2.3.4}$ receiving three queries. In fact, this behavioural pattern can be easily extended such that $\ns{1.2.3.4}$ receives $n$ queries for any $n \in \nat$, by adding more records of the form $\tup{net.attack.nxns,\scrsf{ns},net.victim.nx\{i\}}$ to the zone $attack.net$. \Cref{subfig:nxns-run} illustrates the formal run in the process $\rec{r}$.

The PPT languages required to compute the abstraction are $L_0 = net\cdot attack\cdot nxns \cdot \dnsAlph^+$, and $L_i = net\cdot victim\cdot nx\{i\}\cdot \dnsAlph^+$, where $\dnsAlph = \{net, attack, nxns, 1, 2, 3, other\}$ and $i \in \{1, 2, 3\}$. We  compute the monoid $\scrsf{syn}\subsub{\scrsf{nxns}} = \prod_{0 \leq j \leq 3}\tup{\synmon{L_j}}$, denoting the sytactic domain congruence classes and the corresponding homomorphism $\synhom{\scrsf{nxns}}$. Fixing the packet amplification factor to be $3$ and the target process to be $1.2.3.4$, the formula $\phi_{amp}(3, 1.2.3.4)$ can then be rewritten as

\begin{equation*}\label{eqn:amp-prop}
\resizebox{\textwidth}{!}{$\phi_{amp}'(3,1.2.3.4)  = \exists x_1,\ldots,x_{3}\in\nat.\, \left(\bigwedge\limits_{i\neq j \in \{1,\ldots ,3 \}} x_i \neq x_j\right)
	 \wedge \left( \bigwedge\limits_{i\in \{1,\ldots , 3 \}}\bigvee\limits_{j \in \{r,5.6.7.8\}}\bigvee\limits_{\stackrel{m \in \scrsf{syn}\subsub{\scrsf{nxns}},}{t \in \scrsf{type}}} (\ch_{(j,1.2.3.4)}?\tup{m,t})(x_i)\right)$},
\end{equation*}

\noindent which no longer contains any quantifications other than the positions in the trace. The zone can be translated as shown in \Cref{subfig:nxns-zone}.
\end{example}
}

\paragraph{Unbounded Analysis}
Note that, if there are multiple nameservers, unlike the trivial example given above, there is no limit to the length of runs which can ensure the presence or absence of an amplification attack. Indeed, a resolver may communicate with the same benign nameserver owning multiple zones before establishing communication with the target nameserver. Hence, a complete search for amplification attacks necessitates the unbounded reachability analysis of $\scrsf{dns}$.

\begin{figure}
\label{fig:blackholing}
    \subfloat[\new{An Illustration of rewrite blackholing.}]{
    \label{subfig:blackholing-illustration}
    \centering
   \resizebox{0.55\textwidth}{!}{\includegraphics[width=\textwidth]{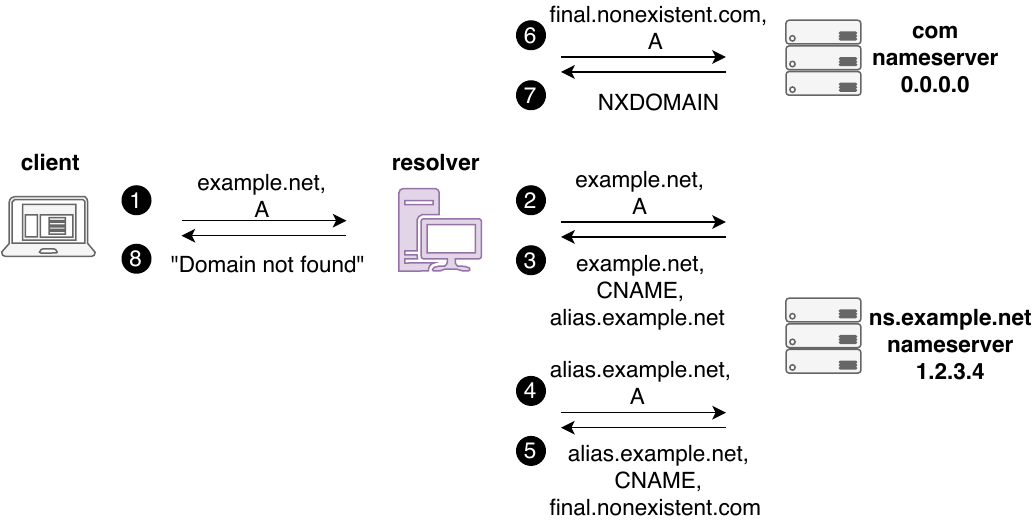}}
    }
\hfill
    \subfloat[\new{The concrete zone configuration.}]{
    \label{subfig:blackholing-zone}
    \centering
   \resizebox{0.4\textwidth}{!}{
    \begin{tikzpicture}[every label/.style={font=\small}]
            \node (before) [fill=blue!5,draw,rounded corners=2pt] at (0,0) {
                \begin{tabular}{l} 
                    --- \textbf{zone} \emph{net.example} ---\\
                    $\tup{net.example,\scrsf{ns},net.example.ns}$\\
                    $\tup{net.example.ns,\scrsf{a},1.2.3.4}$\\
                    $\tup{net.example,\scrsf{cname},net.example.alias}$\\
                    $\tup{net.example.alias,\scrsf{cname},com.nonexistent.final}$\\
                \end{tabular} };
    \end{tikzpicture}}
    }
\hfill
    \subfloat[\new{A local run of the process $\rec{r}$ that corresponds to Figure \ref{subfig:blackholing-illustration}. In the initial configuration, the server list $sl_0 = (net.example.ns)\odot (com)$ contains the domain names of the other two nameservers, and the initial cache $cache_0 = \{\tup{net.example.ns,\scrsf{a},1.2.3.4},\tup{com,\scrsf{a},0.0.0.0}\}$ contains the IP addresses of both these nameservers. We follow the notational conventions from the previous example (see \Cref{subfig:nxns-run}).
    }]{
    \label{subfig:blackholing-run}
    \begin{minipage}[b]{\linewidth}
    \centering
       \includegraphics[width=\textwidth, page=2]{img/formal-run-blackholing-crop.pdf}
    \end{minipage}
    }
    \captionsetup{skip=5pt}
\caption{\new{Rewrite blackholing.}}
\end{figure}

\subsection{Rewrite Blackholing}\label{subsec:rewrite-black}

\new{
\emph{Rewrite blackholing} occurs when a query is eventually redirected to a non-existent domain name. 
This can happen when the resolver, for example, follows some \rtype{cname} records, but ultimately finds that the final ``canonical'' name does not exist.
While rewrite blackholing can be useful for blocking malicious traffic or preventing attacks, it can also be exploited to blackhole legitimate services or domains, thus
 causing a DoS for valid queries.
}

\new{We now show that this attack vector can also be specified in our framework. Informally, given a ZFC $\tup{\cal Z, \Theta}$, the corresponding instance of DNS exhibits rewrite blackholing if there exists some query that, during its resolution process, causes the resolver to eventually receive an answer with the label $\scrsf{nxdomain}$ along with a non-empty set of records. 

We can again formally characterize the set of traces  exhibiting rewrite blackholing as the formula 
\begin{equation*}\label{eqn:rb-prop}
	\phi_{rb} = \exists x\in \nat.\exists recs \in \cal P(\scrsf{record}\subsub{\cal L}). \, \left(recs \neq \emptyset \wedge \bigvee_{j \in P\subsub{\cal Z}} \left((\ch_{(j,p)}?\tup{\scrsf{nxdomain},recs})(x_i)\right) \right).
\end{equation*}}

\new{
\begin{example}
As illustrated in Figure \ref{subfig:blackholing-illustration}, 
a DNS query for \name{example.net} returns a \rtype{cname} pointing to \name{alias.example.net} (\ding{184}). The resolver then queries \name{alias.example.net}, which returns another \rtype{cname} pointing to \name{final.nonexistent.com} (\ding{186}). Since \name{final.nonexistent.com} does not exist, the resolver eventually receives an \rtype{nxdomain} response, indicating that the domain cannot be resolved (\ding{188}). 
In this rewrite blackholing scenario,
two authoritative nameservers are involved: 
one  for handling the resolution of   
\name{example.net} and \name{alias.example.net}
(as both are within the same domain zone \name{example.net}), and the other for \name{nonexistent.com}. The concrete zone file for the \name{example.net} domain is shown in Figure~\ref{subfig:blackholing-zone} and the corresponding formal run in $\rec{r}$ is illustrated in \Cref{subfig:blackholing-run}.
\end{example}
}

	\section{Related Work} 
\label{sec:related}

Our framework, along with its underlying formalism, analysis approach, and application domain, is closely related to the following three lines of  earlier research.

\paragraph{Formal Modeling and Analysis of DNS}
Kakarla et al.~\cite{groot} formalize the DNS resolution semantics, abstracting away critical resolver-side logic including caching and recursive queries. 
Building upon this semantics, they also 
provide the first static analyzer GRoot for identifying DNS misconfiguration errors. 
Liu et al.~\cite{dnsmaude} complement GRoot's semantics by
providing the first formal 
semantics of end-to-end name resolution.
The accompanying tool DNSMaude
offers a collection of formal analyses for both qualitative (e.g., rewrite blackholing)
and quantitative 
properties (e.g., amplification).
We have already discussed
their limitations in \Cref{sec:intro}.

\paragraph{Testing, Monitoring, and Verification of DNS Implementations} 
Recent years have seen significant advances in DNS testing tools~\cite{Eywa,ResolFuzz,ResolverFuzz,scale}. 
ResolFuzz~\cite{ResolFuzz} is a differential tester, searching for semantic
bugs such as inaccuracies in resolver responses across DNS resolver implementations.
ResolverFuzz~\cite{ResolverFuzz} fuzzes resolvers via
 short query-response sequences. 
 It also leverages differential testing and can identify resolver vulnerabilities like cache poisoning bugs.
Unlike pure fuzzing,
SCALE~\cite{scale}
guides its test generation by  
symbolically executing GRoot's  semantics. 
It focuses on detecting RFC compliance errors in authoritative nameserver implementations rather than recursive resolvers.
While also deriving exhaustive test cases using symbolic execution, Eywa~\cite{Eywa} leverages large language models to construct intended DNS behaviors from natural language sources like RFCs. 
All these testers, along with excellent DNS monitors~\cite{thousandeyes,check-host}, 
can only show the presence of bugs or vulnerabilities, not their absence, in DNS implementations.
Moreover, in contrast to our framework, they
do not
encompass the full name resolution process involving both  recursive resolvers and nameservers.

IRONSIDES~\cite{Ironsides} is a nameserver
implementation provably secure against DNS vulnerabilities like single-packet DoS
attacks. 
However, it may still be susceptible to application-layer issues such as amplification attacks, in contrast to our framework. 
DNS-V~\cite{DNS-V} is a verification
framework for an in-house DNS authoritative 
engine in Alibaba Cloud.
A symbolic-execution-based approach is used to 
facilitate the adoption of layered verification, where 
source code within each decomposed layer of DNS can be verified independently with respect to correctness properties. 
Quantitative properties like amplification are out of its scope.
Wang et al.~\cite{10.1145/3663408.3663412}  scale up the verification of DNS configuration
by parallelizing the analysis of zone-file-specific query processing behaviour and symbolically combining the analysis results. 
Unlike our framework, the use of formal methods in none of these works is related to the DNS semantics.

All these testing, monitoring, and verification efforts focus on DNS implementations.
Given that many DNS vulnerabilities have stemmed
 from RFCs or protocol-level defects in heavily tested production resolvers and nameservers, protocol-level verification
 at an early design stage is highly desirable. Our framework takes a step forward towards this goal.

\paragraph{Verification Problem for Infinite-State Systems}
Verification of infinite-state systems is a well-studied problem and is still an active area of research. Abdulla \& Jonsson \cite{AbJo:lossy,AbJo:lossy:IC} and C{\'e}c{\'e} et al. \cite{cece1996unreliable} independently show that the verification problem for CPSes communicating over lossy channels is decidable in non-primitive recursive time, for arbitrary topologies. Abdulla et al. \cite{abdulla2012timed} later show that the decidability holds even when the communicating processes are timed and the messages contain time-stamps. Heussner et al. \cite{heussner2010reachability,heussner2012reachability} study this problem for rCPSes, however, over perfect channels. They restrict their attention to eager runs and show that the problem is decidable in exponential-time for a restricted class of topologies that they call \emph{non-confluent}. Prior to this work, Atig et al. \cite{atig2008reachability} show that the problem is decidable even without the restriction of eagerness, but instead requires the underlying topology to be acyclic. In the context of PDSes, Bouajjani et al. \cite{bouajjani1997reachability} show that the problem is decidable for PDS, using a \textit{saturation method} that computes the set of all reachable computations. Abdulla et al. \cite{abdulla2012dense} and Clemente and Lasota \cite{clemente2015timed} extend the decidability to timed PDS with a timed stack. In contrast, our work handles a timed extension of rCPS with an infinite message alphabet. In addition, we focus on models pertaining to the DNS semantics.

Kakarla et al. \cite{dns-complex} study the expressive complexity of DNS. They show that the verification of DNS zone files, a simpler problem than $\verif$, is likely to take at least cubic time in
the number of records, effectively giving a lower bound. In this work, we provide an upper bound for this problem (\Cref{lem:complexity}).

	\section{Limitations and Conclusion}

A formal framework with a decision procedure for the \scrsf{dnsverif} problem 
is essential to achieving a secure, reliable DNS infrastructure.
Aside from being of theoretical interest in the context of infinite-state systems, 
this work provides a foundation for the advancement of automated verification tools for \emph{and} eliminates the costly $(\text{Break and Fix})^*$ cycle with DNS security analysis.

\new{
We have demonstrated our framework for two applications. 
However,
in principle, we can model all the static vulnerabilities considered by GRoot \cite[Table 1]{groot}, including rewrite blackholing (Section~\ref{subsec:rewrite-black}), rewrite loop, and answer inconsistency.  
This is immediate as every property defined in GRoot is an $\mathsf{LTL}$-property, which is equivalent to $\mathsf{FO}(<)$ over traces.
Additionally, our framework is capable of modeling all the attacks considered by DNSMaude~\cite[Table 4]{dnsmaude}, including those covered by GRoot and amplification attacks (Section~\ref{subsec:amplification}). However, our current model also presents several limitations.
}

\new{\paragraph{Multiple Resolvers}
We are currently limited to having a single resolver in our model. A promising direction for future work would be to consider multiple concurrent resolvers in the system. While our abstraction on the query space would still be sound and complete, the underlying topology would then have two pushdown systems (the resolvers) communicating with each other over a lossy bag channel, which violates the confluence criterion from \cite{heussner2010reachability,heussner2012reachability}.}

\new{\paragraph{Dynamic Updates}
We have adopted the semantics considered by DNSMaude, which subsumes that of GRoot. In particular, our framework currently does not support the query type $\scrsf{IXFR}$, meant for incremental zone transfers. This type is responsible for updating the DNS records stored by a nameserver. Under some conditions, such queries force two-way reliable (TCP) communication between nameservers (which are finite-state systems themselves). The resulting network topology now allows a non-lossy cycle and thus yields the undecidability of the corresponding state-reachability problem, which in turn  implies the undecidability of the $\verif$ problem. More details on this kind of behaviour can be found in \cite{rfc1995}.

Further model-specific investigation is needed to remove the above limitations, which is outside the scope of this paper. Finally, the $\mathsf{2ExpTime}$ complexity upper bound of our procedure is rather pessimistic. Hence,  empirical results are necessary to prove the practicality of our approach. At this point, our contribution is primarily theoretical, and significant engineering work will be required to design and implement these techniques. Conducting an experimental evaluation of our approach remains a challenging future task.}

	\begin{acks}
    We appreciate the  reviewers for their valuable feedback. This research is supported by armasuisse Science and Technology and an ETH Zurich Career Seed Award. We would also like to thank Narayan Kumar and Prakash Saivasan for some useful initial discussions.
\end{acks}

	\clearpage
	\bibliographystyle{ACM-Reference-Format}
	\bibliography{ref,nobi}

\clearpage
\appendix

\section{Transitions in $\scrsf{dns}$}\label{appx:transitions}

\subsection{Authoritative Nameserver Transitions\label{appx:ns-full-trans}}

\begin{figure*}[h]
	$\scrsf{trans}_{s}:=$
	
	\begin{gather}
		\begin{align*}
			\tup{init,\bot,\bot} & \xrightarrow{c_{r,(dn,s)}?q}\tup{lookup,q,r}\thinspace\thinspace\thinspace\thinspace\thinspace\thinspace\thinspace\thinspace\thinspace\thinspace\thinspace\thinspace\thinspace\thinspace\thinspace\thinspace\thinspace\thinspace\thinspace\thinspace\thinspace\thinspace\thinspace\thinspace\thinspace\thinspace\thinspace\thinspace\thinspace\thinspace\thinspace\thinspace\thinspace\thinspace\thinspace\thinspace\thinspace\thinspace\thinspace\thinspace\thinspace\thinspace r\in\mathcal{R},q\in\scrsf{query},dn\in\scrsf{domain}\\
			\tup{lookup,q,r} & \begin{cases}
				\xrightarrow{in}\tup{exact,q,r} & \text{dn}({\cal M}(s,q))=\text{dn}(q)\\
				\xrightarrow{in}\tup{wildcard,q,r} & \text{dn}(q)\in_{*}\text{dn}({\cal M}(s,q))\\
				\xrightarrow{c_{(dn,s),r}!\tup{\scrsf{ansq},\{rr,rr'\}}}\tup{init,\bot,\bot} & \text{dn}({\cal M}(s,q))<\text{dn}(q),\scrsf{dname}\in\tau(s,q)\\
				& \text{where, }\\
				& dn\in\scrsf{domain}\\
				& {\cal T}_{\scrsf{dname}}(s,q)=\{rr\}\\
				& rr'=\tup{\text{dn}(q),\scrsf{cname},\text{ttl}(rr),d,1}\\
				& d=\text{dn}(q)_{|q|}\cdots\text{dn}(q)_{|\text{dn}(rr)|+1}\cdot\text{val}(rr)\\
				\xrightarrow{c_{(dn,s),r}!\tup{\scrsf{ref},{\cal G}(s,q)}}\tup{init,\bot,\bot} & \text{dn}({\cal M}(s,q))<\text{dn}(q),dn\in\scrsf{domain},\scrsf{dname}\not\in\tau(s,q),\scrsf{ns}\in\tau(s,q),\scrsf{soa}\not\in\tau(s,q)\\
				\xrightarrow{c_{(dn,s),r}!\tup{\scrsf{nxdomain},\emptyset}}\tup{init,\bot,\bot} & \text{otherwise for }dn\in\scrsf{domain}
			\end{cases}\\
			\tup{exact,q,r} & \begin{cases}
				\xrightarrow{c_{(dn,s),r}!\tup{\scrsf{ans},{\cal T}(s,q)}}\tup{init,\bot,\bot} & \text{ty}(q)=\scrsf{cname}\in\tau(s,q),\\
				& \scrsf{ns}\not\in\tau(s,q)\vee\scrsf{soa}\in\tau(s,q),|{\cal T}(s,q)|=1\\
				& dn\in\scrsf{domain}\\
				\xrightarrow{c_{(dn,s),r}!\tup{\scrsf{ans},{\cal T}(s,q)}}\tup{init,\bot,\bot} & \scrsf{cname}\neq\text{ty}(q)\in\tau(s,q),\\
				& \scrsf{ns}\not\in\tau(s,q)\vee\scrsf{soa}\in\tau(s,q)\\
				& dn\in\scrsf{domain}\\
				\xrightarrow{c_{(dn,s),r}!\tup{\scrsf{ansq},{\cal T}_{\scrsf{cname}}(s,q)}}\tup{init,\bot,\bot} & \text{ty}(q)\not\in\tau(s,q),\scrsf{cname}\in\tau(s,q),\\
				& \scrsf{ns}\not\in\tau(s,q)\vee\scrsf{soa}\in\tau(s,q),{\cal T}_{\scrsf{cname}}(s,q)=\{rr\}\\
				& dn\in\scrsf{domain}\\
				\xrightarrow{c_{(dn,s),r}!\tup{\scrsf{ref},{\cal G}(s,q)}}\tup{init,\bot,\bot} & \scrsf{ns}\in\tau(s,q),\scrsf{soa}\not\in\tau(s,q),dn\in\scrsf{domain}\\
				\xrightarrow{c_{(dn,s),r}!\tup{\scrsf{ans},\emptyset}}\tup{init,\bot,\bot} & \text{otherwise for }dn\in\scrsf{domain}
			\end{cases}\\
			\tup{wildcard,q,r} & \begin{cases}
				\xrightarrow{c_{(dn,s),r}!\tup{\scrsf{ans},R}}\tup{init,\bot,\bot} & \text{ty}(q)\in\tau(s,q),dn\in\scrsf{domain}\\
				& \text{where, }\\
				& R={\cal T}_{\text{ty}(q)}(s,q)\cup\{\tup{\text{dn}(q),\text{ty}(rr),\text{ttl}(rr),\text{val}(rr),1}\\
				& \mid rr\in{\cal T}_{\text{ty}(q)}(s,q)\}\\
				\xrightarrow{c_{(dn,s),r}!\tup{\scrsf{ansq},\{rr,rr'\}}}\tup{init,\bot,\bot} & \text{ty}(q)\not\in\tau(s,q),\scrsf{cname}\in\tau(s,q),dn\in\scrsf{domain}\\
				& \text{where, }\\
				& {\cal M}(s,q)=\{rr\}\\
				& rr'=\tup{\text{dn}(q),\text{ty}(rr),\text{ttl}(rr),\text{val}(rr),1}\\
				\xrightarrow{c_{(dn,s),r}!\tup{\scrsf{ans},\emptyset}}\tup{init,\bot,\bot} & \text{otherwise for }dn\in\scrsf{domain}
			\end{cases}
		\end{align*}
	\end{gather}
	
	\caption{\label{fig:ns-transitions-old}Transitions in $\ns{s}$}
\end{figure*}

\newpage

\subsection{Recursive Resolver Transitions\label{appx:rec-full-trans}}

\begin{figure}[h]
	$\scrsf{trans}_{r}\coloneqq$
	
	\begin{gather}
		\begin{align*}
			\tup{init,sl} & \xrightarrow{\tup{\epsilon,in,q}}\tup{query,sl}\\
			\tup{query,sl} & \begin{cases}
				\xrightarrow[\nu(rr) > 0]{\tup{q_1\,\,\bottom,in,\varepsilon}} & \tup{init,sl},\\
				& \text{if }\text{dn}(q_{1})\leftarrow_{\text{ty}(q_{1})}rr\\
				\xrightarrow[\nu(rr) > 0]{\tup{q_1q_2,\ch_{r,\text{val}(rr)}!q_{2},\varepsilon}} & \tup{wait,sl},\\
				& \text{if }\text{dn}(q_{1})\leftarrow_{\scrsf{a}}rr\\
				\xrightarrow[\nu(rr) > 0]{\tup{q_1,in,\tup{d,\text{ty}(q_1)}}} & \tup{query,sl},\\
				& \text{if }\text{dn}(q_{1})\leftarrow_{\scrsf{dname}}rr,rr\in cache\\
				& d=\text{dn}(q_{1})_{|q_{1}|}\cdots\text{dn}(q_{1})_{|\text{dn}(rr)|+1}\cdot\text{val}(rr)\\
				\xrightarrow[\nu(rr) > 0]{\tup{q_1, \ch_{r,\text{val}(rr)}!q_{1}, q_1}} & \tup{wait,sl},\\
				& \text{if }ns\leftarrow_{\scrsf{a}}rr\\
				& ns = sl[min_{sl[i] < \text{dn}(q_1)} i]\\
				\xrightarrow{\tup{q_1, in, \tup{ns,\scrsf{a}} \cdot q_1}} & \tup{query,sl},\\
				& \text{if }\{rr\mid ns\leftarrow_{\scrsf{a}}rr,rr\in cache\}=\emptyset\\
				& ns = sl[min_{sl[i] < \text{dn}(q_1)} i]\\
				\xrightarrow{\tup{q_1, \ch_{r,s}!q_1, q_1}} & \tup{wait,sl},\\
				& \text{if }sl = [], s \in \Theta
			\end{cases}\\
			\tup{wait,sl} & \begin{cases}
				\xrightarrow[{[(x_{rr})_{rr\in R}\leftarrow (\text{ttl}(rr))_{rr\in R}]}]{\tup{q_1, \ch_{s,r}?\tup{\scrsf{ans},R},q_1}} & \tup{query,sl}\\
				\xrightarrow[{[(x_{rr})_{rr\in R}\leftarrow (\text{ttl}(rr))_{rr\in R}]}]{\tup{q_1,\ch_{s,r}?\tup{\scrsf{ansq},R},q_1}} & \tup{query,sl}\\
				\xrightarrow[{[(x_{rr})_{rr\in R}\leftarrow (\text{ttl}(rr))_{rr\in R}]}]{\tup{q_1,\ch_{s,r}?\tup{\scrsf{ref},R},q_1}} & \tup{query,list \odot sl},\\
				& list \in Lin(Val(R))\\
				\xrightarrow{\tup{q_1, \ch_{s,r}?\tup{\scrsf{nxdomain},\emptyset}, q_1}} & \tup{query,sl[1\ldots |sl|]}
			\end{cases}
		\end{align*}
	\end{gather}
	\begin{align*}
		\leftarrow_{{\cal T}} & \coloneqq\{(dn,rr)\mid\text{dn}(rr)\leq dn,\text{ty}(rr)\in{\cal T}\}\subseteq\scrsf{domain}\times\scrsf{record}\\
		Val(R) & \coloneqq \{\text{val}(rr)\mid rr\in R,\text{ty}(rr)=\scrsf{ns}\} \subseteq \scrsf{domain}\\
		Lin(V) & \coloneqq \{m_1\ldots m_{|V|} \in V^{|V|} \mid m_i \neq m_j, 0 \leq i < j < |V| \}
	\end{align*}
	
	\caption{Transitions in $\rec{r}$}
\end{figure}

\newpage

\section{Strong Bisimulation}\label{appx:strong-bisim}

For the sake of completeness, we provide the definition of a strong bisimulation on an LTS, as stated in \cite{milner1989communication}.

\begin{definition}\label{def:strong-bisim}
	
	Given an LTS $\cal A^1 = \tup{S^1,S^1_{init},A^1,\trans^1}$, an equivalence relation $R \subseteq S^1 \times S^1$ on $\cal A^1$ is called a \emph{strong bisimulation} on $\cal A^1$ if and only if, for every tuple $\tup{p,p'}\in R$ and all actions $a \in A^1$, if $p \xrightarrow{a} q \in \trans^1$, then $p' \xrightarrow{a} q' \in \trans^1$, such that $\tup{q,q'}\in R$.
	
\end{definition}

It is straightforward to see that, if a generalised kernel bisimulation $H$ on the LTS $\cal A^1$ is such that $h_{act} = id_{A^1}$ is the identity map on $A^1$, i.e., every action label is only related to itself, then the relation $\ker{h_{state}}$ is a strong bisimulation.

\end{document}